\documentclass[10pt]{amsart}
\title[Set-theoretic YB, reflection equations $\&$ quantum groups]{ Set-theoretic Yang-Baxter 
$\&$  reflection equations and quantum group symmetries}
\author[Anastasia Doikou and Agata Smoktunowicz]{Anastasia Doikou and Agata Smoktunowicz}

\address[A. Doikou] {Department of Mathematics, Heriot-Watt University,
Edinburgh EH14 4AS, and The Maxwell Institute for Mathematical Sciences, Edinburgh}
\email{A.doikou@hw.ac.uk}

\address[A. Smoktunowicz] {
School of Mathematics, The University of Edinburgh, The Kings Buildings, Mayfield Road, Edinburgh EH9 3JZ, 
and The Maxwell Institute for Mathematical Sciences, Edinburgh}
\email{A.Smoktunowicz@ed.ac.uk}

\usepackage[english]{babel}
\usepackage[utf8]{inputenc}
\usepackage{amsmath, amssymb, bm}
\usepackage{graphicx}
\usepackage{subcaption}
\usepackage[colorinlistoftodos]{todonotes}
\usepackage{indentfirst}

\newcommand{\hiddenpower}[2] { \ifnum \numexpr#2=1 #1 \else #1^#2 \fi }
\numberwithin{equation}{section}

\newcommand{\cal}{\mathcal}

\usepackage{xparse}
\ExplSyntaxOn
\newcounter{diff_order}
\newcounter{diff_power}

\newcommand{\rawdiff}[3]
{
	\setcounter{diff_order}{0}
	\clist_map_inline:nn{#3}{\stepcounter{diff_order}}
	
	\frac{\hiddenpower{#1}{\thediff_order} #2}
	{
		\def\old_var{DefaultValue}
		\setcounter{diff_power}{0}
		
		\clist_map_inline:nn{#3}
		{
			\def\new_var{##1}
			\ifnum \thediff_power=0
				\stepcounter{diff_power}
			\else
				\tl_if_eq:NNTF \new_var \old_var
				{\stepcounter{diff_power}}
				{
					#1 \hiddenpower{\old_var}{\thediff_power}
					\setcounter{diff_power}{1}
				}
			\fi

			\def\old_var{##1}
		}
		
		#1 \hiddenpower{\old_var}{\thediff_power}
	}
}
\ExplSyntaxOff

\newlength{\bibitemsep}
\newlength{\bibparskip}
\let\oldthebibliography\thebibliography
\renewcommand\thebibliography[1]{%
  \oldthebibliography{#1}%
  \setlength{\parskip}{\bibitemsep}%
  \setlength{\itemsep}{\bibparskip}%
}

\newtheorem{thm}{Theorem}[section]
\newtheorem{lemma}[thm]{Lemma}
\newtheorem{cor}[thm]{Corollary}
\newtheorem{pro}[thm]{Proposition}
\newtheorem{defn}[thm]{Definition}

\newtheorem{rem}[thm]{Remark}

\begin{document}

%\begin{center}
%\strut\hfill

%\noindent\rule{14cm}{0.8pt}

%\vskip 0.45in

%\noindent {{\bf{SET-THEORETIC SOLUTIONS OF YANG-BAXTER AND REFLECTION EQUATIONS $\&$ QUANTUM GROUP SYMMETRIES}}}\\
%\vskip 0.3in
%
%\noindent {{\footnotesize ANASTASIA DOIKOU AND AGATA SMOKTUNOWICZ}}
%\vskip 0.4in

%\vskip 0.60in

%\end{center}

\begin{abstract}
\noindent  Connections between set-theoretic Yang-Baxter and reflection equations and quantum integrable systems
are investigated. We show that set-theoretic $R$-matrices are expressed as  twists of known solutions.
We then focus on reflection and twisted algebras and we derive the associated defining algebra relations for $R$-matrices being  
Baxterized solutions of the $A$-type Hecke algebra ${\cal H}_N(q=1)$.
We show in the case of the reflection algebra that there exists a ``boundary'' finite sub-algebra for some special choice of  
``boundary'' elements of the $B$-type Hecke algebra ${\cal B}_N(q=1, Q)$. We also show the key proposition that the associated 
double row transfer matrix is  essentially 
 expressed  in terms of the elements of the $B$-type Hecke algebra. This  is one of the  fundamental results of this investigation
together with the proof of the duality between the boundary finite subalgebra  and the $B$-type Hecke algebra.
These are  universal statements that  largely generalize  previous relevant findings, and also allow 
the investigation of the symmetries of the double row transfer matrix.\\
\\
{\it 2010 MCS}:  16T20, 16T25, 17B37, 20G42, 82B23.
\end{abstract}

\maketitle

\date{}
\vskip 0.4in

%\newpage

%\tableofcontents

\section{Introduction}

\noindent The Yang-Baxter equation and the $R$-matrix are central objects in the framework of quantum integrable systems.
 The Yang-Baxter equation  was  first introduced by Yang in \cite{Yang}
when investigating  many particle systems with  $\delta$-type  interactions
and later in the celebrated work of Baxter,  who solved the anisotropic Heisenberg magnet (XYZ model) \cite{Baxter}.
The solution of the model by Baxter was achieved by implementing the so-called Q-operator method, a sophisticated
approach leading to sets of functional relations known as T-Q relations, that provide information on the spectrum of the model. 
A different approach on the resolution of the spectrum of 1D statistical models is the Quantum
Inverse Scattering (QISM) method, an elegant algebraic technique \cite{Korepin}, that
led directly to the invention of quasitriangular Hopf algebras known as quantum groups, which then formally developed  
by Jimbo and Drinfeld independently \cite{Drinfeld, Jimbo}. 

Drinfield  \cite{Drin} also suggested  the idea of  set-theoretic solutions to the Yang-Baxter equation, and since then 
a lot of research activity has been devoted to this issue (see for instance \cite{Hienta},  \cite{ESS}).  
Set-theoretical solutions and Yang-Baxter maps have been investigated in the context of classical discrete integrable systems
related also to the notion of  Darboux-B{\"a}cklund transformation \cite{ABS, Veselov, Papag}. 
Links between the set-theoretical Yang-Baxter equation and geometric crystals \cite{Eti, BerKaz}, 
or soliton cellular automatons \cite{TakSat, HatKunTak}  have been also revealed.
Set-theoretical solutions of the Yang-Baxter equations have been investigated by employing
the theory of  braces and skew-braces. The theory of braces  was established by  
W. Rump who developed a structure called a brace  to describe 
all finite involutive  set-theoretic solutions of  the Yang-Baxter equation \cite{[25], [26]}.   
He showed that every brace provides a solution to the Yang-Baxter equation, and every 
non-degenerate, involutive set-theoretic  solution of the Yang-Baxter equation can be obtained from a brace, 
a structure that generalizes nilpotent rings. Skew-braces were then developed in \cite{GV}  
to describe non-involutive solutions. Key links between  set-theoretical solutions and quantum integrable systems and the associated 
quantum algebras were uncovered in \cite{DoiSmo}.

Following the works of Cherednik \cite{Cherednik} and Sklyanin \cite{Sklyanin}, 
who introduced and studied the reflection equation,  
much attention has been focused on the issue of incorporating boundary conditions to integrable
models. The boundary effects, controlled by the refection equation, shed new light on the bulk theories themselves, and also paved the way
to new mathematical concepts and physical applications. The set-theoretical reflection equation together with the first examples of solutions first 
appeared in \cite{CauZha}, while a more systematic study and a classification inspired by maps appearing in integrable discrete systems presented 
in \cite{CauCra}. Other solutions were also considered and used 
within the context of cellular automata \cite{KuOkYa}. In \cite{SmoVenWes, Katsa} methods coming from the theory of braces were used  
to produce families of new solutions to the reflection equation, and in \cite{DeCommer} skew braces were used to
produce reflections.

\vskip 0.15cm

\noindent {\bf The outline of the paper.} In this study  we consider set-theoretic solutions of
the Yang-Baxter and reflection equations coming from braces and we construct quantum
spin chains with open boundary conditions through Sklyanin's double row transfer matrix \cite{Sklyanin}. 
We should mention that typical well studied solutions of the Yang-Baxter equation are the Yangians, expressed as 
$R(\lambda)={\cal P}+\lambda I$, where ${\cal P}$ is the  flip map: $u \otimes v \mapsto  v\otimes u$. 
Here we consider more general classes of solutions of the Yang-Baxter equation that are expressed as $R(\lambda)={\cal P}+\lambda {\cal P}{\check r}$,
where ${\check r}$ is a map that can be obtained for instance from a brace. Such solutions are of particular interest, given that in general they have no semi-classical analogue and as such they are distinctly different from the known quantum group solutions. Let us  describe below in more detail what is achieved in each section:
\begin{itemize}
\item In section $2$ we present some basic background information.  
More precisely, in subsection 2.1 we review  some background on $R$-matrices associated to non-degenerate, involutive, 
set-theoretic solutions of the Yang-Baxter equation as well as set-theoretic solutions of the reflection equation and some information on braces.
Then in subsection 2.2  we provide a review on recent results on the connections of brace solutions of the Yang-Baxter equation and 
the corresponding quantum algebras and integrable quantum spin chains \cite{DoiSmo}. 
\item
In section 3 examples of set-theoretic $R$-matrices expressed as simple twists of known solutions 
via isomorhisms  within the finite set $\{1, \ldots, {\cal N}\}$ are presented.
Based on these solutions we construct explicitly the  
associated ``twisted'' co-products by employing the finite set isomorphisms. 
%An intermezzo on the $q$-deformed case then follows, where similar simple twists on the familiar representation of the $A$-type 
%Hecke algebra ${\cal H}_N(q)$ \cite{Jimbo2}
%are performed. which are concerned with a 
%special class of solutions known as Lyubashenko's solutions \cite{Drin}. 
We then move on to show 
that the generic brace solution of the Yang-Baxter equation can be obtained
from the permutation operator via suitable Drinfeld twists \cite{Drinfeldtw}. Note that the properties of the brace structures are instrumental 
in deriving the form of the twist.
Certain generalizations regarding the $q$-deformed case are also discussed.

\item In section $4$ we focus on quadratic algebras, i.e. the reflection and twisted algebras \cite{Sklyanin, Olsha}.  

\begin{enumerate}
\item In subsections 4.1 and 4.2 we review some background information on reflection algebras and $B$-type Hecke algebras. 
More precisely, in subsection 4.1. 
we recall the links  between the refection algebras and $B$-type Hecke algebras and the Baxterization process, 
whereas in subsection 4.2 we discuss set-theoretic representations of the $B$-type Hecke algebra by essentially reviewing 
some recent results on solutions of the set-theoretic reflection equation \cite{SmoVenWes}. 

\item In subsection 4.3 we derive the associated defining algebra relations for Baxterized solutions of the $A$-type Hecke algebra ${\cal H}_{N}(q=1)$,
and we show in the case of the 
reflection algebra that there exist a finite sub-algebra for some special choice of  ``boundary'' elements of the $B$-type Hecke algebra,
which also turns out to be a symmetry of the double row transfer matrix for special boundary conditions as will be shown in subsection 5.2.
\end{enumerate}

\item In section $5$ we introduce open spin chains like systems and we focus on the study of the
associated quantum group symmetries.  We first review the construction of open quantum spin chains via the use of 
tensorial representations of the reflection algebras and  the derivation of the double row transfer matrix. 
The findings of each subsection are described below.

\begin{enumerate}

\item In subsection 5.1.we study the symmetries of the double row transfer matrix constructed from Baxterized solutions of the $B$-type Hecke algebra 
${\cal B}_{N}(q=1, Q)$ . We first  prove the key proposition of this study, i.e. we show 
that almost all the factors, but one, of the $\lambda$-series expansion  of the open transfer matrix can
be expressed in terms of the elements of the $B$-type Hecke algebra. Interestingly, when choosing special boundary conditions, 
the full open transfer matrix can be exclusively expressed in terms of elements of the $A$-type Hecke algebra.
Another fundamental result is that that all elements of the of the $B$ type Hecke algebra ${\cal B}_N(1=1, Q)$ commute  
with a finite sub algebra of the reflection algebra. 
This then leads to another important proposition regarding the symmetry of the associate double row transfer matrix.
These are  universal results that largely extends earlier partial findings  (see e.g. \cite{PasqSal, DoikouMurphy}),  
and are of particulate physical and mathematical significance.

\item In subsection 5.2 more symmetries of open transfer matrices associated to certain classes of set-theoretic
solutions of the Yang-Baxter equation  coming from braces are also discussed. The derivation of these symmetries is primarily based
on the properties of the brace structures.  
Some of these symmetries  generalize recent findings on periodic  transfer matrices \cite{DoiSmo}, while others are new.

\item In subsection 5.3 symmetries of the double row transfer matrix constructed from the special class of 
Lyubashenko's solutions are identified confirming  also some of the findings of section 3.
\end{enumerate}

\end{itemize}

% Section 5.
%This section contains the most fundamental results in our analysis.
%b.The main result symmetries of double transfer matrix-general case(obtained by Anastasia 
%with difficult proof): plan of this section 

%(i) recalling classical results for classical integrable systems about block of monodromy matrices 
%commuting with transfer matrix 
%
%(ii) showing the classical results hold for our systems, so that (i) holds for our systems. 
%
%(iii) the main result of the section-that symmetries of double transfer matrix can be presented using 
%$r_{i,i+1} $ and k (b element) 
%
%5(iv) Examples of symmetries obtained from the main result 

\section{Preliminaries}

\noindent We present in this section some basic background information regarding set-theoretic solutions of the 
Yang-Baxter and reflection equations and braces as well as a brief review on the recent findings of \cite{DoiSmo}
on the links between set-theoretic solutions of the Yang-Baxter equation from braces and quantum algebras.

\subsection{The set-theoretic Yang-Baxter equation}
\noindent Let $X=\{{\mathrm x}_{1}, \ldots, {\mathrm x}_{{\cal N}}\}$ be a set and ${\check r}:X\times X\rightarrow X\times X$.
 Denote \[{\check r}(x,y)= \big (\sigma _{x}(y), \tau _{y}(x)\big ).\] 
We say that $r$ is non-degenerate if $\sigma _{x}$ 
and $\tau _{y}$ are bijective functions. Also, the solutions $(X, \check r)$  is involutive: $\check r ( \sigma _{x}(y), \tau _{y}(x)) = 
(x, y)$, ($\check r \check r (x, y) = (x, y)$). We focus on non-degenerate, involutive solutions of the set-theoretic braid equation:
\[({\check r}\times id_{X})(id_{X}\times {\check r})({\check r}\times id_{X})=(id_{X}\times {\check r})({\check r}
\times id_{X})(id_{X}\times {\check r}).\]

%Suppose that $(X, {\check r})$ is an involutive, non-degenerate set-theoretic solution of the Braid equation:
%\[({\check r}\times Id_{X})(Id_{X}\times {\check r})({\check r}\times Id_{X})=(Id_{X}\times {\check r})({\check r}\times Id_{X})(Id_{X}\times {\check r}).\]  
%et $r=P {\check r}$ be the corresponding solution of the quantum Yang-Baxter equation (where $P$ is the permutation matrix). 
%Then the check- matrix related  to $(X,{\check r})$ is ${\check r}={\check r}_{i,j; k,l}$ where ${\check r}_{i,j; k,l}=1 $ if and 
%only if ${\check r}(i,j)=(k,l)$, and is zero otherwise. 
%Notice that ${\check r}^{2}=I$ the identity matrix, because $\check r$ is involutive.
%Let $n$ be a natural number and let $e_{i,j}$ denote the $n\times n$ matrix whose entries are $0$ except for the $i,j$-th entry which is $1$.

Let $V$ be the space of dimension equal to the cardinality of $X$, and with a slight abuse of notation, 
let  $\check r$ also  denote the $R$-matrix associated to the linearisation of ${\check r}$ on 
$V={\mathbb C }X$ (see \cite{LAA} for more details), i.e.
$\check r$  is the ${\cal N}^2 \times{\cal N}^2$ matrix: 
\begin{equation}
\check r = \sum_{x,y,z,w \in X} \check r(x,z|y,w) e_{x,z} \otimes e_{y, w}, 
\end{equation}
where
$e_{x, y}$ is the ${\cal N} \times {\cal N}$  matrix: $(e_{x,y})_{z,w} =\delta_{x,z}\delta_{y,w} $.
Then for the $\check r$-matrix related  to $(X,{\check r})$:  ${\check r}(x,z|y,w)=\delta_{z, \sigma_x(y)} \delta_{w, \tau_y(x)}$.
 Notice that the matrix $\check r:V\otimes V\rightarrow V\otimes V$ satisfies the (constant) Braid equation:
\[({\check r}\otimes I_{V})(I_{V}\otimes {\check r})({\check r}\otimes I_{V})=(I_{V}\otimes {\check r})({\check r}\otimes 
I_{V})(I_{V}\otimes {\check r}).\]
Notice also that ${\check r}^{2}=I_{V \otimes V}$ the identity matrix, because $\check r$ is involutive.

For set-theoretical solutions it is thus convenient to use the matrix notation:
\begin{equation}
{\check r}=\sum_{x, y\in X} e_{x, \sigma_x(y)}\otimes e_{y, \tau_y(x)}. \label{brace1}
\end{equation}
Define also, $r={\cal P}{\check r}$, where ${\cal P} = \sum_{x, y \in X} e_{x,y} \otimes e_{y,x}$ 
is the permutation operator,  consequently
${ r}=\sum_{x, y\in X} e_{y,\sigma_x(y)}\otimes e_{x, \tau_y(x)}.$
The  Yangian is a special case: $\check r(x,z|y,w)= \delta_{z,y}\delta_{w,x} $.

$ $

Let $ (X, \check r)$ be a non-degenerate set-theoretic solution to the Yang-Baxter equation. 
A map $k : X\to X$ is a reflection of $(X, \check r)$ if it satisfies 
\begin{equation*}
		\check r(k\times id_X) \check r (k\times id_X)=(k\times id_X) \check r(k\times id_X) \check r.
	\end{equation*}
We say that $k$ is a set-theoretic solution to the reflection equation. We also say that  $k$ is involutive if $k(k(x))  = x$.

Using the matrix notation introduced above then the reflection matrix $K$ is  an ${\cal N} \times {\cal N}$ matrix
represented as:
\begin{equation}
{\mathrm k} = \sum_{x\in X} e_{x, k(x)}
\end{equation}
and satisfies the constant reflection equation:
\begin{equation}
\check r(  {\mathrm k}  \otimes I_V)  \check r ( {\mathrm k}\otimes  I_V )=({\mathrm k}\otimes  I_V ) 
\check r({\mathrm k}\otimes  I_V ) \check r.
\end{equation}

$ $
Let us now recall the role of braces in the derivation of set-theoretic solutions of the Yang-Baxter equation.
In \cite{[25], [26]} Rump showed that every solution $(X, \check r)$ can be in a good way embedded in a brace.
\begin{defn}[Proposition $4$, \cite{[26]}]
A {\em left brace } is an abelian group $(A; +)$ together with a 
multiplication $\cdot $ such that the circle operation $a \circ  b =
 a\cdot b+a+b$ makes $A$ into a group, and $a\cdot (b+c)=a\cdot b+a\cdot c$.
\end{defn} 
 In many papers, an equivalent definition is used \cite{[6]} .
% \begin{defn}[\cite{[6]}]
% A {\em left brace } is a set $G$ together with binary operations $+$ and $\circ $ such that 
% $(G, +)$ is an abelian group, $(G, \circ )$ is a group, and 
%  $a\circ (b+c)+a=a\circ b+a\circ  c$ for all $a,b,c\in G$.
%\end{defn} 
 The additive identity of a brace $A$ will be denoted by $0$ and the multiplicative identity by $1$. In every brace $0=1$. 
The same notation will be used for skew braces (in every skew brace $0=1$).
 
%Some authors use the notation $\cdot $  instead of $\circ $ and $*$ instead of $\cdot $ (see for example \cite{[6], gateva, Gateva}).

%In \cite{ESS}, Etingof, Schedler and Soloviev introduced  the retract relation for any  solution $(X,r)$. Denote  $X=\{x_{1}, \ldots , x_{\mathcal N}\}$ and
%$r(x,y)=(\sigma _{x}(y), \tau _{y}(x))$.  
%Recall that the retract relation $\sim $ on $X$  is defined by $x_{i}\sim   x_{j}$  if
%$\sigma _{i} = \sigma _{j}$. The  induced solution
%$Ret(X, r) = (X/\sim , r^{\sim })$  is called the {\em retraction } of $X$. A solution $(X, r)$ is
%called  a {\em  multi-permutation solution of level m} if m is the smallest non-negative
%integer such that  after $m$ retractions we obtain the solution with one element.

Throughout this paper we will use the following result,  which  is  implicit in \cite{[25], [26]} 
and explicit in Theorem 4.4 of \cite{[6]}.

\begin{thm}\label{Rump}(Rump's theorem, \cite{[25], [26], [6]}).  It is known that for an involutive, non degenerate solution  
of the braid equation there is always an underlying brace $(B, \circ , +)$, 
such that the maps $\sigma _{x}$ and $\tau _{y}$ come from this brace, and $X$ is a subset in this brace such that ${\check r}(X,X)
\subseteq (X,X)$ and ${\check r}(x,y)=(\sigma _{x}(y), \tau _{y}(x))$, where $\sigma _{x}(y)=x\circ y-x$, $\tau _{y}(x)=t\circ x-t$, 
where $t $ is the inverse of $\sigma _{x}(y)$ in the circle group $(B, \circ )$.  Moreover, we can assume that every element from $B$ 
belongs to the additive group $(X, +)$  generated by elements of $X$. In addition every solution of this type is a non-degenerate, involutive 
set-theoretic solution of the braid equation.
\end{thm}

 We will call the brace $B$  an underlying brace of the solution $(X,{\check r})$, or a brace associated to the  
solution $(X,{\check r})$. We will also say that the solution $(X, \check r )$ is associated
 to brace $B$. Notice that this is also related  to the formula of
 set-theoretic solutions associated to the braided group (see \cite{ESS} and \cite{gateva}). 

The following remark was also discovered by Rump.
 
\begin{rem} \label{nilpotent}
 Let $(N, +, \cdot)$ be an associative ring which is a nilpotent ring. For $a,b\in N$ define 
 $\ a\circ b=a\cdot b+a+b$, then $\ (N, +, \circ )$ is a brace.
 \end{rem}

%\begin{defn}
% Let $X, Y$ be sets and ${\check r}:X\times X\rightarrow X\times X$, ${\check r}':Y\times Y\rightarrow Y\times Y$ be functions.%
%Let $(X, {\check r})$ and $(Y,{\check r}')$ be set-theoretic  solutions of the Braid equation, and left $f:X\rightarrow Y$ be a function onto $X$ such that 
%${\check r}'(f(x), f(y))=(f\times f)({\check r}(x,y))$,  for all $x,y\in X$. Then $f$ is called a homomorphism of solutions.%
%If $f$ is $one-to-one$ then $f$ is called an isomorphism of solutions.
% \end{defn}

\subsection{Yang Baxter equation $\&$ quantum groups}

\noindent In this subsection we briefly review the main results reported in \cite{DoiSmo} on the various links between braces, 
representations of the $A$-type Hecke algebras, 
and quantum algebras.

Recall first the Yang-Baxter equation in the braid form ($\delta = \lambda_1 - \lambda_2$):
\begin{equation}
\check R_{12}(\delta)\ \check R_{23}(\lambda_1)\ \check R_{12}(\lambda_2) = \check R_{23}(\lambda_2)\
 \check R_{12}(\lambda_1)\ \check R_{23}(\delta) . \label{YBE1}
\end{equation}
We focus here on brace  solutions\footnote{All, finite, non-degenerate, involutive, 
set-theoretic solutions of the YBE (\ref{brace1}) are coming from braces (Theorem \ref{Rump}),
therefore we will call such solutions {\it  brace solutions}.}  of the YBE, given by (\ref{brace1}) and the Baxterized solutions: 
\begin{equation}
\check R(\lambda) = \lambda \check r + {\mathbb I}, \label{braid1}
\end{equation}
where ${\mathbb I}= I_X\otimes I_X$  and $I_{X}$ is the identity matrix of dimension equal to the cardinality of the set $X$.
Let also $R = {\cal P} \check R$, (recall the permutation operator ${\cal P} = \sum_{x, y\in X} e_{x,y}\otimes e_{y,x}$),  
then the following basic properties for $R$ matrices coming from braces were shown in \cite{DoiSmo}:\\

\noindent {\bf Basic Properties.} {\em  
The brace $R$-matrix satisfies the following fundamental properties:}
\begin{eqnarray}
&&  R_{12}(\lambda)\  R_{21}(-\lambda) = (-\lambda^2 +1) {\mathbb I}, 
~~~~~~~~~~~~~
\mbox{{\it Unitarity}} \label{u1}\\
&&  R_{12}^{t_1}(\lambda)\ R_{12}^{t_2}(-\lambda -{\cal N}) = 
\lambda(-\lambda -{\cal N}){\mathbb I}, 
~~~~~\mbox{{\it Crossing-unitarity}} \label{u2}\\
&& R_{12}^{t_1 t_2}(\lambda) = R_{21}(\lambda), \label{tt}\nonumber
\end{eqnarray}
{\it where $^{t_{1,2}}$ denotes transposition on the fist, second space respectively, 
and recall ${\cal N}$ is the same as the cardinality of the set  $X$.}

Let us also recall the connection of the brace representation with the $A$-type Hecke algebra.

\begin{defn} The $A$-type Hecke algebra ${\cal H}_N(q)$ is defined by 
the generators $g_l$, $l \in \{1,\ 2, \ldots, N-1 \}$ and the exchange relations:
\begin{eqnarray}
&& g_l\ g_{l+1} \ g_l = g_{l+1}\ g_l\ g_{l+1}, \label{h1} \\
&& \Big [ g_l,\ g_m \Big ]= 0, ~~~|l-m|>1 \label{h2} \\
&& \big ( g_l-q \big )\big (g_l+q^{-1} \big) =0. \label{h3} 
\end{eqnarray}
\end{defn}

\begin{rem} The brace solution $\check r$ (\ref{brace1})
is a representation of the $A$-type Hecke algebra for $q=1$.
Indeed, $\check r$ satisfies the braid relation and $\check  r^2 =1,$ 
which is shown by using the involution property. Also, the braid relation is  
satisfied by means of the brace properties (see also Theorem \ref{Rump} and \cite{DoiSmo}).
%Then via Baxerization (see e.g. \cite{Jimbo}) the solution of the braid 
%Yang Baxter equation is given by (\ref{braid1}).
\end{rem}

\subsection*{The Quantum Algebra associated to braces}
\noindent 
Given a solution of the Yang-Baxter equation, the quantum algebra is defined via the fundamental relation \cite{FadTakRes}
(we have multiplied the familiar RTT relation with the permutation operator):
\begin{equation}
\check R_{12}(\lambda_1 -\lambda_2)\ L_1(\lambda_1)\ L_2(\lambda_2) = L_1(\lambda_2)\ L_2(\lambda_1)\ 
\check R_{12}(\lambda_1 -\lambda_2). \label{RTT}
\end{equation}
$\check R(\lambda) \in \mbox{End}({\mathbb C}^{{\cal N}} \otimes {\mathbb C}^{{\cal N}})$, $\ L(\lambda) \in 
\mbox{End}({\mathbb C}^{{\cal N}}) \otimes {\mathfrak A}$, where ${\mathfrak A}$\footnote{Notice that in $L$ 
in addition to the indices 1 and 2 in (\ref{RTT}) there is also an implicit ``quantum index`$n$ associated to ${\mathfrak A},$ 
which for now is omitted, i.e. one writes $L_{1n},\ L_{2n}$.} is the quantum algebra defined by (\ref{RTT}). 
We shall focus henceforth on solutions associated to braces only given by (\ref{braid1}), (\ref{brace1}). 
The defining relations of the coresponding 
quantum algebra were derived in \cite{DoiSmo}:\\

\noindent {\it The quantum algebra associated to the brace $R$ matrix  (\ref{braid1}), (\ref{brace1}) 
is defined by generators $L^{(m)}_{z,w},\ z, w \in X$, and defining relations }
\begin{eqnarray}
L_{z,w}^{(n)} L_{\hat z, \hat w}^{(m)} - L_{z,w}^{(m)} L_{\hat z, \hat w}^{(n)} &=& 
L^{(m)}_{z, \sigma_w(\hat w)} L^{(n+1)}_{\hat z,\tau_{\hat w}( w)}- L^{(m+1)}_{z, \sigma_w(\hat w)} 
L^{(n)}_{\hat z, \tau_{\hat w}( w)}\nonumber\\ &-& L^{(n+1)}_{ \sigma_z(\hat z),w} 
L^{(m)}_{\tau_{\hat z}( z), \hat w }+ L^{(n)}_{ \sigma_z(\hat z, )w}  L^{(m+1)}_{\tau_{\hat z}( z), \hat w}. \label{fund2}
\end{eqnarray}

The proof is based on the fundamental relation (\ref{RTT}) and the form of the brace $R$- matrix (for the detailed proof see \cite{DoiSmo}). Recall also that
in the index notation we define $\check R_{12} = \check R \otimes \mbox{id}_{\mathfrak A}$:
\begin{eqnarray}
&&  L_1(\lambda) = \sum_{z, w \in X} e_{z,w} \otimes I \otimes L_{z,w}(\lambda),\ \quad  L_2(\lambda)= \sum_{z, w \in X}
I  \otimes  e_{z,w}  \otimes L_{z,w}(\lambda).  \label{def}
\end{eqnarray} 
%where $L_{z,w}(\lambda)$ 
%are the generators of the affine algebra ${\mathfrak A}$ and $\check R$ is given in (\ref{braid1}), (\ref{brace1}). 
The exchange relations among the various generators of the affine algebra 
are derived below via (\ref{RTT}). Let us express $L$ as a formal power series expansion 
$L(\lambda) = \sum_{n=0}^{\infty} {L^{(n)} \over \lambda^n}$.
Substituting  expressions (\ref{braid1}), and the $\lambda^{-1}$ expansion in (\ref{RTT}) we obtain
the defining relations of the quantum algebra associated 
to a brace $R$-matrix (we focus on terms $\lambda_1^{-n} \lambda_2^{-m}$):
\begin{eqnarray}
&&  \check r_{12} L_{1}^{(n+1)} L_2^{(m)} -\check  r_{12} L_1^{(n)} L_2^{(m+1)} +  L_1^{(n)} L_2^{(m)} \nonumber\\
&&  = L_1^{(m)} L_{2}^{(n+1)} \check r_{12} -  L_1^{(m+1)} L_2^{(n)}\check r_{12} +  L_1^{(m)} L_2^{(n)}. \label{fund}
\end{eqnarray}
The latter relations immediately lead to the quantum algebra relations (\ref{fund2}), after recalling:
$
L_{1}^{(k)}=\sum_{x,y\in X}e_{x,y}\otimes I \otimes  L_{x,y}^{(k)},$ $ L_{2}^{(k)}=\sum_{x,y\in X} I \otimes 
e_{x,y}\otimes  L^{(k)}_{x,y}, \nonumber
$
and $\check r_{12} = \check r \otimes  \mbox{id}_{\mathfrak A },$ 
$L^{(k)}_{x,y} $ are the generators of the associated quantum algebra.
The quantum algebra is also equipped with a co-product $\Delta: {\mathfrak A} \to {\mathfrak A} \otimes {\mathfrak A}$ \cite{FadTakRes, Drinfeld}. Indeed, we define 
\begin{equation}
{\mathrm T}_{1;23}(\lambda)= L_{13}(\lambda) L_{12}(\lambda),\  
\end{equation}
which satisfies (\ref{RTT}) and is expressed as ${\mathrm T}_{1;23}(\lambda) = \sum_{x,y \in X} e_{x,y} \otimes \Delta(L_{x,y}(\lambda)).$
%as well as an anti-pode $s:  {\mathfrak A} \mapsto {\mathfrak A}$, $(\mbox{id} \otimes s) L(\lambda) = L^{-1}(\lambda)$ 
%and a co-unit $\epsilon:  {\mathfrak A} \mapsto {\mathfrak A}$,
%$(\mbox{id} \otimes \epsilon) L(\lambda)= \mbox{id}$.

\begin{rem} In the special case $\check r ={\cal P}$ the ${\cal Y}(\mathfrak {gl}_{\cal N})$ algebra is recovered:
\begin{equation}
\Big [ L_{i,j}^{(n+1)},\ L_{k,l}^{(m)}\Big ] -\Big [ L_{i,j}^{(n)},\ L_{k,l}^{(m+1)}\Big ] = L_{k,j}^{(m)}L_{i,l}^{(n)}- L_{k,j}^{(n)}L_{i,l}^{(m)}. \label{fund2b}
\end{equation}
\end{rem}

The next natural step is  the classification of solutions of the fundamental relation (\ref{RTT}), 
for the brace quantum algebra.  A first step towards this goal  will be to examine  the fundamental object $L(\lambda)= L_0 + {1\over\lambda} L_1$, 
and search for finite and infinite representations of the respective elements. The fusion procedure \cite{Fusion} can be also exploited to yield 
higher dimensional representations of the associated quantum algebra.
The classification of $L$-operators  will allow the identification of  new classes of quantum  integrable systems, such as analogues of Toda chains or deformed boson models.
A first obvious example to consider is associated to Lyubashenko's solutions, which are discussed in what follows.
 This is a significant direction to pursue and will be systematically addressed elsewhere.

\section{Set-theoretic solutions as Drinfeld twists}
\noindent  In this section we first introduce some special cases of solutions of the braid equation that are immediately 
obtained from fundamental known solutions. We  show in particular that a special class of solutions known as 
Lyubashenko's solutions \cite{Drin} can 
be expressed as simple twists.
Although the construction is simple it has significant implications on the associated symmetries of the braid solutions. 
We then move on to show 
that the generic brace solution of the Yang-Baxter equation (\ref{brace1}) can be obtained
from the permutation operator via a suitable Drinfeld twist \cite{Drinfeldtw}, and we identify the specific form of the twist.
Moreover, inspired by the isotropic case we provide a  similar construction for the $q$-deformed analogue of Lyubashenko's solution.

Before we derive the Lyubashenko solution  as a suitable twist we first introduce a useful Lemma.

\begin{lemma} {\label{extra1}} Let $\check r': V \otimes V \to V \otimes V$ ($V$ is a finite dimensional space) 
satisfy the braid relation and $(\check r')^2 = I^{\otimes 2}$. Let also ${\mathbb  V}: V \to V$ be an invertible map,
such that $({\mathbb V} \otimes {\mathbb V})  \check r' = \check r'  ({\mathbb V} \otimes {\mathbb V})$. 
We define $\check r = ({\mathbb V} \otimes I)\check  r' ({\mathbb V}^{-1} \otimes I) =
(I\otimes {\mathbb V}^{-1})\check r' (I \otimes{\mathbb V}),$ then: 
\begin{enumerate}
\item ${\check r}^2 = I^{\otimes 2}$
\item $\check r$ satisfies the braid relation.
\end{enumerate}
\end{lemma}

\begin{proof}

\begin{enumerate}

\item 
$\check r^2 = ({\mathbb V} \otimes I) (\check  r')^2 ({\mathbb V}^{-1} \otimes I) = ({\mathbb V} \otimes I) I^{\otimes 2} ({\mathbb V}^{-1} \otimes I) = I^{\otimes 2}.$

\item It is given that $\check r'$ satisfies the braid relation:
\begin{equation}
(\check r' \otimes I) (I \otimes \check  r')( \check r' \otimes I) = (I \otimes \check  r')( \check r' \otimes I)(I \otimes \check  r'). \label{br22}
\end{equation}
We express: $ \check r' \otimes I =  ({\mathbb V}^{-1} \otimes I \otimes I) ( \check r \otimes I  ) ({\mathbb V} \otimes I \otimes I) $ and 
$I \otimes \check r' = (I \otimes I \otimes {\mathbb V}) (I \otimes \check r) (I \otimes I \otimes {\mathbb V}^{-1}).$
We then conclude for the left hand side of (\ref{br22}):
\begin{eqnarray}
%&&  ({\cal V}^{-1} \otimes I \otimes I) (\check r \otimes I)  ({\cal V} \otimes I \otimes I)  (I \otimes I \otimes {\cal V})(I \otimes \check  r) (I \otimes I \otimes {\cal V}^{-1}) ({\cal V}^{-1} \otimes I \otimes I)( \check r \otimes I) ({\cal V} \otimes I \otimes I) \nonumber\\
LHS:\   ({\mathbb V}^{-1} \otimes I \otimes {\mathbb V}) (\check r\otimes I) (I \otimes \check  r)( \check r \otimes I)  ({\mathbb V} \otimes I \otimes {\mathbb V}^{-1}).\label{A}
\end{eqnarray}
Similarly, for the RHS of (\ref{br22}) :
\begin{equation}
  RHS:\ ({\mathbb V}^{-1} \otimes I \otimes {\mathbb V}) (I \otimes \check  r)( \check r \otimes I)(I \otimes \check  r) ({\mathbb V} \otimes I \otimes {\mathbb V}^{-1}). \label{B}
\end{equation}
From (\ref{A}), (\ref{B}) we conclude that $\check r$ also satisfies the braid relation.
\end{enumerate}
\end{proof}

\begin{pro}\label{prop1} Let $\tau,\ \sigma: X \to X$, $X= \{1,\ldots, {\cal N} \}$  be isomorphisms, 
such that $\sigma(\tau(x)) = \tau(\sigma(x)) = x$ and let
${\mathbb V}=\sum_{x \in X} e_{x, \tau(x)}$ and ${\mathbb V}^{-1} = \sum_{x \in X}  e_{ \tau(x),x}$.
Then any solution of the type (Lyubashenko's solution)
\begin{equation}
\check r=\sum_{x, y \in X} e_{x, \sigma(y)} \otimes e_{y, \tau(x)}, \label{special1}
\end{equation}
can be obtained from the permutation operator ${\cal P}= \sum_{x, y \in X} e_{x,y} \otimes e_{y,x}$ as
\begin{equation}
\check r = ({\mathbb V}\otimes I ){\cal P} ({\mathbb V}^{-1} \otimes I)= ( I \otimes {\mathbb V}^{-1} ) {\cal P} (I \otimes {\mathbb V} ) \label{special1b}
\end{equation}
\end{pro}
\begin{proof}
The proof relies on the definitions 
of ${\cal P},\ {\mathbb V},\ {\mathbb V}^{-1}$ and the fundamental property $e_{x,y} e_{z,w} = \delta_{y,z} e_{x,w}$:
\begin{eqnarray}
&& ({\mathbb V}\otimes I ){\cal P} ({\mathbb V}^{-1} \otimes I) = \nonumber \\
&& (\sum_{z\in X} e_{\sigma(z), z} \otimes I) (\sum_{x,y \in X} e_{x,y} \otimes 
e_{y,x}) (\sum_{w\in X} e_{w, \sigma(w)} \otimes I)= \nonumber\\
&&  \sum_{x,y \in X} e_{\sigma(x), \sigma(y)} \otimes e_{y,x}.\nonumber
\end{eqnarray}
The latter is indeed equal to (\ref{special1}) given that $\sigma(\tau(x)) = \tau (\sigma(x)) = x$, and due to Lemma \ref{extra1}, 
$\check r$ (\ref{special1}) satisfies the braid relation and $\check r^2 =I^{\otimes 2}$.
\end{proof}

Note that $r = {\cal P} \check r$,  and consequently $R = {\cal P} \check R$
 take a simple form for this class of solutions:
\begin{equation}
r= {\mathbb V}^{-1} \otimes {\mathbb V}\ \Rightarrow\  R(\lambda) = \lambda {\mathbb V}^{-1} \otimes {\mathbb V} + {\cal P}. \label{special2}
\end{equation}

\noindent {\bf Examples}:

1. $\sigma(y) =y+1,\ \tau(x) =x-1$,  (see also \cite{LAA}).

2. $\sigma(y) ={\cal N} +1-y,\ \tau(x) ={\cal N}+1-x$.

Note that in both examples above $x, y \in \{ 1, \ldots, {\cal N}\}$ and $\sigma, \tau$ in example 1 are defined $\mbox{mod}{\cal N}$.

Before we present our findings on the symmetry of Lyubashenko's $\check r$-matrix we first introduce a useful Lemma.

\begin{lemma}{\label{extra2}} Let ${\mathfrak l}_{x,y}$ be the generators of the $\mathfrak{gl}_{\cal N}$ algebra satisfying:
\begin{equation}
\Big [{\mathfrak l}_{x,y},  {\mathfrak l}_{z,w}\Big ] = \delta_{y,z}{\mathfrak l}_{x,w} - \delta_{x,w}{\mathfrak l}_{z,y}. \label{gl2}
\end{equation}
The  $\mathfrak{gl}_{\cal N}$ algebra is equipped with a coproduct $\Delta: \mathfrak{gl}_{\cal N} \to  \mathfrak{gl}_{\cal N}  \otimes  \mathfrak{gl}_{\cal N}$ such that
\begin{equation}
\Delta({\mathfrak l}_{x,y}) = {\mathfrak l}_{x,y} \otimes \mbox{id} +  \mbox{id}  \otimes {\mathfrak l}_{x,y}.
\end{equation}
The $N$-coproduct is obtained by iteration $ \Delta^{(N)} =  ( \Delta^{(N-1)}  \otimes \mbox{id})  \Delta= (\mbox{id} \otimes \Delta^{(N-1)}) \Delta $ and is given as $\Delta^{(N)}({\mathfrak l_{x,y}})=  \sum_{n=1}^N \mbox{id} \otimes \ldots \otimes\underbrace{{\mathfrak l}_{x,y}}_{n^{th}\  \mbox{position}} \otimes \ldots \otimes \mbox{id}$.

Let also ${\mathfrak F}^{(N)}: \mathfrak{gl}_{\cal N}^{\otimes N} \to \mathfrak{gl}_{\cal N}^{\otimes N}$ 
be an invertible element such that  ${\mathfrak F}^{(N)}\Delta^{(N)}({\mathfrak l}_{x,y}) = \Delta_T^{(N)}({\mathfrak l}_{x,y}){\mathfrak F}^{(N)},$ then $\Delta_T^{(N)}({\mathrm l}_{x,y})$ also satisfy the $\mathfrak{gl}_{\cal N}$ algebraic relations.
\end{lemma}
\begin{proof} The $N$-coproducts satisfy the $\mathfrak{gl}_{\cal N}$ relations (\ref{gl2}), i.e. 
$\Big [\Delta^{(N)}({\mathfrak l}_{x,y}),  \Delta^{(N)}({\mathfrak l}_{z,w})\Big ] = \delta_{y,z}\Delta^{(N)}({\mathfrak l}_{x,w}) - \delta_{x,w}\Delta^{(N)}({\mathfrak l}_{z,y}).$ By acting from the left with ${\mathfrak F}^{(N)}$ and with $({\mathfrak F}^{N})^{-1}$ from the right in the latter commutator we immediately obtain $\Big [\Delta_T^{(N)}({\mathfrak l}_{x,y}),  \Delta_T^{(N)}({\mathfrak l}_{z,w})\Big ] = \delta_{y,z}\Delta_T^{(N)}({\mathfrak l}_{x,w}) - \delta_{x,w}\Delta_T^{(N)}({\mathfrak l}_{z,y})$.
\end{proof}

\begin{cor}\label{prop2} {\it Let  $\rho: \mathfrak{gl}_{\cal N} \to \mbox{End}({\mathbb C}^{\cal N})$ be the fundamental representation of $\mathfrak{gl}_{\cal N},$  such that ${\mathfrak l}_{x,y} \mapsto e_{x,y},$ where recall $e_{x, y}$ are ${\cal N} \times {\cal N}$ matrices with elements 
$(e_{x,y})_{z,w}=\delta_{x,z} \delta_{y,w}.$ The special solution $\check r$ (\ref{special1}) is $\mathfrak{gl}_{\cal N}$ symmetric, i.e.}
\begin{equation}
\Big [ \check r,\ \Delta_i(e_{x,y}) \Big ] =0, ~~~~x,\ y \in X, \label{symm1}
\end{equation}
{\it where we define the ``twisted'' co-products ($i= 1, 2$): }
\begin{eqnarray}
&&\Delta_1(e_{x,y}) = e_{\sigma(x), \sigma(y)} \otimes I +I  \otimes e_{x,y}, \nonumber\\
&& \Delta_2(e_{x,y}) =  e_{x,y} \otimes  I + I \otimes   e_{\tau(x) ,\tau(y)} ,  \label{symm2}
\end{eqnarray}
($\Delta_1(e_{\tau(x), \tau(y)}) = \Delta_2(e_{x, y})$).
\end{cor}
\begin{proof}
This can be shown using the form of the special class of solutions (\ref{special1}). 
The permutation operator is $\mathfrak{gl}_{\cal N}$ symmetric, i.e.
\begin{equation}
\Big [ {\cal P},\ \Delta(e_{x,y})\Big ] =0, \label{comm1}
\end{equation}
where the co-products $\Delta(e_{x,y})$  are defined in Lemma (\ref{extra2}) (${\mathfrak l}_{x,y} \mapsto e_{x,y}$).
%\begin{eqnarray}
%&& \Delta({\mathrm Y} ) = \mbox{id} \otimes {\mathrm Y} + {\mathrm Y}  \otimes \mbox{id}, 
%~~~~{\mathrm Y}  \in \mathfrak {gl}_{\cal N} \nonumber\\
%&&   \Delta^{(N)}({\mathrm Y})=  \sum_{n=1}^N \mbox{id} \otimes \ldots \otimes\underbrace{{\mathrm Y}}_{n^{th}\  
%\mbox{position}} \otimes \ldots \otimes \mbox{id},%
% \label{coproduct1}
%\end{eqnarray}
%and the element ${\mathrm Y}$ appears in the $n^{th}$ position of the $N$ co-product.
%$e_{x,y}$ and consequently $\Delta(e_{x,y})$  form a basis of $\mathfrak{gl}_{\cal N}$. 

Let ${\mathbb V}= \sum_{x \in X} e_{x, \tau(x)}$, then (\ref{symm1}) immediately follows from (\ref{comm1}) and (\ref{special1b}) after 
multiplying (\ref{comm1}) from the left and right with ${\mathbb V}\otimes I,\  {\mathbb V}^{-1}  \otimes I$  or  
$ I \otimes {\mathbb V}^{-1},\   I \otimes {\mathbb V}$ respectively.  $\Delta_{i}(e_{x,y})$ are then defined as
\begin{eqnarray} 
&& \Delta_1(e_{x,y}) = {\mathbb V} e_{x,y} {\mathbb V}^{-1} \otimes I + I \otimes e_{x,y}, \nonumber\\ 
&& \Delta_2(e_{x,y}) = e_{x,y} \otimes I +I \otimes {\mathbb V}^{-1}   e_{x,y}{\mathbb V} \label{sim1}
\end{eqnarray}
and explicitly given by (\ref{symm2}). Indeed, ${\mathbb V} e_{x,y} {\mathbb V}^{-1} = e_{\sigma(x), \sigma(y)}$ and  ${\mathbb V}^{-1} e_{x,y} {\mathbb V} = e_{\tau(x), \tau(y)}.$

According to Lemma \ref{extra2} $\Delta_i(e_{x, y})$ also satisfy the $\mathfrak{gl}_{\cal N}$ algebra relations, thus $\check r$ (\ref{special1}) is $\mathfrak{gl}_{\cal N}$ symmetric. In this particular case, as is clear from the computation above, two invertible linear maps are involved, ${\cal F}^{(2)}_i: \mbox{End}({\mathbb C}^{\cal N} \otimes {\mathbb C}^{\cal N}) \to  \mbox{End}({\mathbb C}^{\cal N} \otimes {\mathbb C}^{\cal N}), i \in \{ 1, 2\}$ such that ${\cal F}^{(2)}_1 := {\mathbb V} \otimes I$ and ${\cal F}^{(2)}_2 := I \otimes {\mathbb V}^{-1}$ and ${\cal F}^{(2)}_i\Delta(e_{x,y}) = \Delta_i(e_{x,y}) {\cal F}^{(2)}_i.$
\end{proof}

 By iteration one derives the $N$ co-products:
$ \Delta_1^{(N)} =  ( \Delta_1^{(N-1)}  \otimes \mbox{id})  \Delta_1$ and  $\Delta_2^{(N)} = (\mbox{id} \otimes \Delta_2^{(N-1)}) \Delta_2$,
which explicitly read as
\begin{eqnarray}
&& \Delta_1^{(N)}(e_{x,y}) =\sum_{n=1}^N I \otimes \ldots \otimes e_{\sigma^{N-n}(x),\sigma^{N-n}(y)}
\otimes \ldots \otimes I , \label{delta1}\\
&& \Delta_2^{(N)}(e_{x,y}) =\sum_{n=1}^N I \otimes \ldots \otimes e_{\tau^{n-1}(x),\tau^{n-1}(y)}
\otimes \ldots \otimes I, \label{delta2}
\end{eqnarray}
The above expressions can be written in a compact form as:\\ 
$\Delta^{(N)}_i(e_{x,y}) = {\cal F}_i^{(N)}\Delta^{(N)}(e_{x,y})({\cal F}_i^{(N)})^{-1},$ where\\  $\Delta^{(N)}(e_{x,y})=  \sum_{n=1}^N \mbox{id} \otimes \ldots \otimes\underbrace{e_{x,y}}_{n^{th}\  \mbox{position}} \otimes \ldots \otimes \mbox{id},$ and we define 
${\cal F}_{1}^{(N)}:= {\mathbb V}^{N-1} \otimes {\mathbb V}^{N-2} \otimes \ldots \otimes {\mathbb V} \otimes I$ and 
${\cal F}_2^{(N)} := I \otimes {\mathbb V}^{-1} \otimes {\mathbb V}^{-2}\otimes \ldots \otimes {\mathbb V}^{-(N-1)}$ (see also relevant findings in \cite{DoikouDr}).
%\noindent{\bf  Remark 2.} {\it A more general $V$ could have been considered: $V$ The proof goes as above \ref{pro2}, 
%but now the modified co-products read as:
%\begin{eqnarray}
%&&\Delta_1(e_{xy}) =C_{\sigma(x)} \bar C_{\sigma(y)} e_{\sigma(x) \sigma(y)} \otimes \mbox{id} + \mbox{id} \otimes e_{xy}, \nonumber\\
%&& \Delta_2(e_{xy}) =  e_{xy} \otimes  \mbox{id} + \mbox{id} \otimes  C_y \bar C_x  e_{\tau(x) \tau(y)} .  \label{symm2}
%\end{eqnarray}}

It was shown in \cite{DoiSmo} that the periodic Hamiltonian for systems built with $R$-matrices associated 
to the Hecke algebra ${\cal H}_N(q=1)$ 
is expressed exclusively in terms of the $A$-type Hecke algebra elements. In the special case 
where $\check  r = {\cal P}$, i.e. the Yangian  the periodic transfer matrix is $\mathfrak{gl}_{\cal N}$  
symmetric.  However, if we focus on the more general class of Lyubashenko's solutions of  
Proposition \ref{prop1} and Corollary \ref{prop2} 
we conclude that because of the  existence of the term $\check r_{N1}$ (due to periodicity) \cite{DoiSmo}, and also due to the form of 
the modified co-products (\ref{delta1}), (\ref{delta2}), the periodic Hamiltonian and in general the periodic transfer matrix is not 
$\mathfrak{gl}_{\cal N}$ symmetric anymore. 
However, we shall be able to show  in section 5 that for a special choice of boundary 
conditions not only the corresponding Hamiltonian is $\mathfrak{gl}_{\cal N}$ symmetric, but also the double row transfer matrix.
This means that the open spin chain enjoys more symmetry compared to the periodic one similarly to the $q$-deformed case 
\cite{PasqSal, Kulishsym, DoikouNepomechie, Doikou2, CraDoi}.
It is therefore clear that from this point of view open spin chains are rather more natural objects to consider compared to the periodic ones.
In \cite{DoiSmo} a systematic investigation of symmetries of the periodic transfer matrix for generic representations of the $A$-type
Hecke algebra ${\cal H}_{N}(q=1)$ as well as for certain solutions of the Yang-Baxter equation coming from braces is presented.

With the following proposition we generalize the results on Lyubashenko's solutions. Specifically, we express the generic 
 brace $\check r$-matrix (\ref{brace1}) as a twist of the permutation operator. Drinfeld introduced \cite{Drinfeldtw} the``twisting'' (or deformation) 
of a (quasi) triangular Hopf algebra that produces yet another (quasi) triangular (quasi) Hopf algebra (see also relevant
\cite{Kulishtw,Maillet}). Let us briefly recall the notion of a twist. Let $\check R$ be the quantum group 
invariant matrix i.e. it {\it commutes} with
the the respective quantum algebra \cite{Jimbo, ChaPre}.  We are  focusing on the finite  algebra ${\mathfrak g},$  
specifically we are considering here the algebras $\mathfrak{gl}_{\cal N}$ or ${\mathfrak U}_q(\mathfrak{gl}_{\cal N})$, 
although via the evaluation homomorphism  one obtains the corresponding affine algebras, i.e. the Yangian ${\cal Y}(\mathfrak{gl}_{\cal N})$  or the affine ${\mathfrak U}_q(\widehat {\mathfrak{gl}_{\cal N}})$ respectively \cite{Jimbo, ChaPre}.
Consider the fundamental representation $\pi: {\mathfrak g } \to \mbox{End}({\mathbb C}^{{\cal N}})$, the  co-poducts
$\Delta: {\mathfrak g} \to {\mathfrak g} \otimes {\mathfrak g}$ and the $\check R$-matrix satisfy linear intertwining relations:
$(\pi \otimes \pi)\Delta(X)\ \check R = \check R\ (\pi \otimes \pi)\Delta(X)$ for $X \in {\mathfrak g}$.
Let also ${\cal F} \in \mbox{End}({\mathbb C}^{{\cal N}} \otimes {\mathbb C}^{\cal N})$, then the $\check R$ matrix can be 
``twisted'' as ${\cal F} \check R {\cal F}^{-1}$, where ${\cal F}$ also
satisfies a set of constraints dictated by the YBE. Given the linear intertwining relations and the twisted $\check R$-matrix, 
one derives the twisted co-products of the finite algebra as
 ${\cal F}\ (\pi \otimes \pi)\Delta(X)\  {\cal F}^{-1}$ (for a more detailed exposition on the notions of quasi-trinagular 
Hopf algebras and Drinfeld twists the interested reader is referred for instance to  \cite{ChaPre}).

\begin{pro} \label{Drin} Let $\check r = \sum_{x,y \in X} e_{x, \sigma_x(y)} \otimes e_{y, \tau_y(x)}$
 be the brace solution of the Yang-Baxter equation (see also (\ref{brace1}) and footnote 1 in page 6).  
Let also $V_{k}$, $k \in \{1, \ldots, {\cal N}^2 \}$ be the eigenvectors 
of the permutation operator ${\cal P} = \sum_{x,y\in X} e_{x,y} \otimes e_{y,x}$, and  $\hat V_{k}$, 
$k \in \{1, \ldots, {\cal N}^2 \}$ 
be the eigenvectors of the brace 
$\check r$ matrix.
Then the $\check r$ matrix can be expressed as  a Drinfeld twist, such that 
$\check r={\cal F} {\cal P} {\cal F}^{-1}$, where the twist ${\cal F}$ is explicitly expressed as 
${\cal F} = \sum_{k=1}^{{\cal N}^2} \hat V_k\  V_k^T$.
\end{pro}
\begin{proof} We divide our proof in three parts:
\begin{enumerate}
\item First we diagonalize the permutation operator. Let $\hat e_j$ be the ${\cal N}$ dimensional column 
vectors with one at the $j^{th}$ position and zero elsewhere, then  the (normalized) 
eigenvectors of the permutation operator are ($x,\ y \in X$):
\begin{eqnarray}
&& V_{k} = {1\over \sqrt{2}} \big ( \hat e_x \otimes \hat e_y + \hat e_y \otimes \hat e_x \big),   
~~~~~~k \in \big \{1, \ldots, {{\cal N}^2 + {\cal N} \over 2}\big \}, \nonumber\\
&&  V_{k} = {1\over \sqrt{2}} \big ( \hat e_x \otimes \hat e_y - \hat e_y \otimes \hat e_x \big),   
~~~~~~k \in \big \{ {{\cal N}^2 + {\cal N} \over 2}+1, \ldots, {\cal N }^2 \big \}, ~~~~x\neq y. \nonumber
\end{eqnarray}
The first ${{\cal N}^2 + {\cal N} \over 2}$ eigenvectors  have the same eigenvalue $1$, while the
 rest ${{\cal N}^2 - {\cal N} \over 2}$ eigenvectors  have eigenvalue $-1$. Also it is easy to check that $V_k$ 
form an ortho-normal basis for the ${\cal N}^2$ dimensional space. Indeed, 
$V^T_k V_l = \delta_{kl}$ and 
$\sum_{k=1}^{{\cal N}^2} V_{k}V_{k}^T = I_{{\cal N}^2}$ ($^T$ denotes usual transposition).

\item Second we diagonalize the brace $\check r$-matrix. First we observe that 
\begin{eqnarray}
\check r\ e_{x} \otimes e_{y} = e_{\sigma_x(y)} \otimes e_{\tau_y(x)}, ~~~~\check r\ e_{\sigma_x(y)} 
\otimes e_{\tau_y(x)}=  e_{x} \otimes e_{y}. \nonumber
\end{eqnarray}
Then we find that the eigenvectors of the $\check r$ matrix are
\begin{eqnarray}
&& \hat V_{k} = {1\over \sqrt{2}} \big ( \hat e_x \otimes \hat e_y + \hat e_{\sigma_x(y)} \otimes \hat e_{\tau_y(x)} \big),   
~~~~~~k \in \big \{1, \ldots, {{\cal N}^2 + {\cal N} \over 2}\big \}, \nonumber\\
&&  \hat V_{k} = {1\over \sqrt{2}} \big ( \hat e_x \otimes \hat e_y - \hat e_{\sigma_x(y)} \otimes \hat e_{\tau_y(x)}\big),   
~~~~(x,y )\neq (\sigma_x(y), \tau_y(x)), \nonumber \\
&& k \in \big \{ {{\cal N}^2 + {\cal N} \over 2}+1, \ldots, {\cal N }^2 \big \}. \nonumber
\end{eqnarray}
As in the case of the permutation operator the $\check r$ matrix has the same eigenvalues
$1$ and $-1$ and the same multiplicities, 
${{\cal N}^2 + {\cal N} \over 2}$ and  ${{\cal N}^2 - {\cal N} \over 2}$ respectively. 
Hence, the two matrices are similar, i.e. there exists some ${\cal F} \in \mbox{End}({\mathbb C}^{{\cal N}} \otimes 
{\mathbb C}^{\cal N})$ 
(not uniquely defined) such that $\check r = {\cal F} {\cal P} {\cal F}^{-1}$.

\item Our task now is to  derive the explicit form of ${\cal F}$. 
This is quite straightforward, indeed the eigenvalue problem for ${\cal P}$ (and $\check r$) reads as 
\begin{eqnarray}
{\cal P} V_k = \lambda_k V_k\ \Rightarrow \  {\check r} \hat V_k = \lambda_k \hat V_k \nonumber
\end{eqnarray}
where, via  $\check r = {\cal F} {\cal P} {\cal F}^{-1}$, we identify ${\cal F} V_{k } = \hat V_k$,
 which by using $\sum_{k=1}^{{\cal N}^2} V_k V_k^T =I$, leads to the  explicit expression  
${\cal F} = \sum_{k=1}^{{\cal N}^2}\hat V_k\  V_k^T$. 
\end{enumerate}
\end{proof}

Note that  if $\check r = {\cal F} {\cal P} {\cal F}^{-1}$ (${\cal P}$ the permutation operator) then 
$r = {\cal P} \check r = {\cal F}^{(op)} {\cal F}^{-1},$ where ${\cal F}^{(op)} = {\cal P} {\cal F} {\cal P},$ and consequently the Baxterized solution (\ref{braid1}) is given as $R(\lambda) = {\cal P} \check R (\lambda) = \lambda {\cal F}^{(op)} {\cal F}^{-1}+ {\cal P}.$

\begin{cor} The brace solution $\check r$ (\ref{brace1})  is $\mathfrak{gl_{\cal N}}$ symmetric, i.e. 
$\big [\check r,\ \Delta_T(e_{x,y}) \big ] =0$, where the twisted co-products are given as 
$\Delta_T(e_{x ,y}) = {\cal F} \Delta(e_{x, y}){\cal F}^{-1}$.
\end{cor}
\begin{proof}
The proof is straightforward as in Corollary \ref{prop2} 
using the fact that the permutation operator is $\mathfrak{gl_{\cal N}}$ symmetric. 
\end{proof}

Notice that here we identified the Drinfeld twist as a similarity transformation between 
the permutation operator and the brace solution. The twisted $n$-co-product as well as the 
$n$ form of ${\cal F}$ should be identified and the admissibility of the twist should be also examined.  Also, issues on the co-associativity of the 
co-product need to be addressed.  We already observe in the simple case of Lyubashenko’s 
solutions that the co-associativity of the twisted co-products is not guaranteed. 
These are significant issues that are addressed in \cite{DoikouDr}.

\subsection{Parenthesis: the $q$-deformed case}

\noindent We slightly deflect in this subsection from our main issue, 
which is the set-theoretic solutions of the Yang-Baxter equation, 
and briefly discuss the $q$-deformed case. Inspired by the special class of Lyubashenko's solutions, 
we generalize 
in what follows Proposition  \ref{prop1} and  Corollary \ref{prop2} in the case of the 
${\mathfrak U}_q(\mathfrak{gl}_{\cal N})$ invariant representation of the $A$-type Hecke algebra \cite{Jimbo}:
\begin{equation}
{\mathrm g} =\sum_{x \neq y=1}^{\cal  N} \Big ( e_{x, y} \otimes e_{y,x}  - q^{-sgn(x-y)} e_{x,x} \otimes e_{y,y} \Big ) +q. \label{braidq}
\end{equation}
Note that strictly speaking this solution is not a set-theoretic solution of the braid equation. Nevertheless, 
isomorphisms within the set of integers $\{1, \ldots, {\mathcal N}\}$ can be still exploited to yield generalized solutions
 based on (\ref{braidq}).

\begin{pro} \label{Prop3} Let $\sigma,\ \tau: X \to X$ be isomorprhisms ($X = \{1, \ldots, {\cal N} \}$) 
such that $\sigma(\tau(x)) = \tau(\sigma(x)) = x$. The quantity 
\begin{eqnarray}
{\mathrm G} &=& \sum_{x\neq y=1}^{\cal N} \Big (e_{x, y} \otimes e_{\tau(y), \tau(x)}  - 
q^{-sgn(x-y)} e_{x,x} \otimes e_{\tau(y),\tau(y)} \Big )+ q\nonumber\\ 
&=&  \sum_{x\neq y=1}^{\cal N}\Big ( e_{\sigma(x), \sigma(y)} \otimes e_{y, x}  - q^{-sgn(x-y)} e_{\sigma(x),\sigma(x)} 
\otimes e_{y,y} \Big )+ q\label{specialq}
\end{eqnarray}
can be obtained from  the ${\mathfrak U}_q(\mathfrak{gl}_{\cal N})$ 
invariant braid solution (\ref{braidq}), provided that $sgn(x-y) = sgn(\tau(x)-\tau(y)) = sgn(\sigma(x)- \sigma(y))$, 
and is also a representation of the $A$-type Hecke algebra.
\end{pro}
{\begin{proof}
Let ${\mathbb V}= \sum_{w}e_{w,\tau(w)}$, and $~{\mathbb V}^{-1} = \sum_{z}  e_{\tau(z), z}$. 
We show by explicit computation that,
\begin{equation}
({\mathbb V}\otimes {\mathbb V})\ {\mathrm g} = {\mathrm g}\ ({\mathbb V}\otimes {\mathbb V}) \label{symq}
\end{equation}
provided that $sgn(\tau(x)-y) = sgn(\sigma(x)- y)$. We then define, bearing in mind (\ref{symq}):
\begin{equation}
{\mathrm G} = ({\mathbb V} \otimes I)\ {\mathrm g}\  ( {\mathbb V}^{-1}\otimes I)= (I
\otimes {\mathbb V}^{-1})\ {\mathrm g}\  (I \otimes {\mathbb V}), \label{defG}
\end{equation}
which leads to (\ref{specialq}). 

Also, ${\mathrm g}$ is a given representation of the $A$-type Hecke algebra, i.e.
\begin{eqnarray}
&& ( {\mathrm g}\otimes I)\ (I \otimes{\mathrm g})\  ({\mathrm g}\otimes I) =
(I \otimes {\mathrm g})\  ({\mathrm g}\otimes I)\  (I \otimes {\mathrm g}) , \label{h1q}\\
&& \big ( {\mathrm g} -q \big )\big ({\mathrm g}+q^{-1} \big) =0. \label{h3q}
\end{eqnarray}
By multiplying (\ref{h1q}) with ${\mathbb V} \otimes I \otimes {\mathbb V}^{-1}$ from the left and 
${\mathbb V}^{-1} \otimes I \otimes {\mathbb V}$ 
from the right, and also multiplying (\ref{h3q}) 
with  ${\mathbb V} \otimes I$ from the left and  ${\mathbb V}^{-1} \otimes I$ from the right, 
and using the definition 
(\ref{defG}) we immediately conclude that ${\mathrm G}$ is also a representation 
of the $A$-type Hecke algebra (see also Lemma \ref{extra1}). 
\end{proof}

It will be useful for what follows  to recall the basic definitions regarding the 
${\mathfrak U}_q(\mathfrak{gl}_{\cal N})$ algebra \cite{Jimbo}. Let
\begin{equation}
a_{ij} = 2 \delta_{ij} - \delta_{i, j+1}-\delta_{i+1, j}, 
~~~i,\ j\in \{1, \ldots , {\cal N}-1 \} 
\end{equation}
be the Cartan matrix of the associated Lie algebra.

\begin{defn} \label{def3} The quantum 
algebra ${\mathfrak U}_q(\mathfrak{sl}_{\cal N})$ has the Chevalley-Serre generators $e_{i}$, $f_{i}$,
$q^{\pm {h_{i}\over 2}}$, $i\in \{1, \ldots, {\cal N}-1\}$ obeying the defining relations:
\begin{eqnarray}
&& \Big [q^{\pm {h_{i}\over 2}},\ q^{\pm {h_{j}\over 2}} \Big]=0\,  ~~~~
q^{{h_{i}\over 2}}\ e_{j}=q^{{1\over 2}a_{ij}}e_{j}\ q^{{h_{i}\over 2}}\, ~~~~
q^{{h_{i}\over 2}}\ f_{j}= q^{-{1\over 2}a_{ij}}f_{j}\ q^{{h_{i}\over 2}}, \nonumber\\
&& \Big [e_{i},\ f_{j}\Big ] = \delta_{ij}{q^{h_{i}}-q^{-h_{i}} \over q-q^{-1}},
~~~~i,j \in \{ 1, \ldots, , {\cal N}-1 \} \label{1} 
\end{eqnarray}
and the $q$ deformed Serre relations 
\begin{eqnarray} 
&& \sum_{n=0}^{1-a_{ij}} (-1)^{n}
\left [ \begin{array}{c}
  1-a_{ij} \\
   n \\ \end{array} \right  ]_{q} 
\chi_{i}^{1-a_{ij}-n}\ \chi_{j}\ \chi_{i}^{n} =0, ~~~\chi_{i} \in \{e_{i},\ f_{i} \}, ~~~ i \neq j. \label{chev} 
\end{eqnarray}
\end{defn}

\begin{rem} \label{rem3}  $q^{\pm h_{i}}=q^{\pm (\epsilon_{i} -\epsilon_{i+1})}$, where the elements
$q^{\pm \epsilon_{i}}$ belong to ${\mathfrak U}_q(\mathfrak{gl}_{\cal N})$. Recall that ${\mathfrak U}_q(\mathfrak{gl}_{\cal N})$
 is derived by adding to ${\mathfrak U}_q(\mathfrak{sl}_{\cal N})$ 
the elements $q^{\pm \epsilon_{i}}$ $i\in \{1, \ldots, {\cal N}\}$  so that $q^{\sum_{i=1}^{{\cal N}}\epsilon_{i}}$ 
belongs to the center \cite{Jimbo}, and $\big [q^{\epsilon_i},\ q^{\epsilon_j} \big] =0$,  $~q^{\epsilon_{i}}\ e_{j} = q^{\delta_{i,j}- \delta_{i,j+1}} e_j\  q^{\epsilon_i}$, $~q^{\epsilon_{i}}\ f_{j} = q^{-(\delta_{i,j}- \delta_{i,j+1})} f_j\ q^{\epsilon_i}$.
\end{rem}

\noindent ${\mathfrak U}_q(\mathfrak{gl}_{\cal N})$ is equipped with a co-product 
$\Delta:{\mathfrak U}_q(\mathfrak{gl}_{\cal N})
\to {\mathfrak U}_q(\mathfrak{gl}_{\cal N}) \otimes {\mathfrak U}_q(\mathfrak{gl}_{\cal N})$ such that
\begin{equation}
 \Delta(\xi) = q^{- {h_{i} \over 2}} \otimes \xi + \xi \otimes q^{{h_{i} \over 2}}, ~~\xi \in \Big \{e_{i},\ f_{i}\Big  \},
\qquad \Delta(q^{\pm{\epsilon_{i} \over 2}}) = q^{\pm{\epsilon_{i} \over 2}} \otimes
q^{\pm{\epsilon_{i} \over 2}} .\label{cop}
\end{equation} 
%It will be useful to define also %
%$\Delta': {\mathfrak U}_q(\mathfrak{gl}_{\cal N})\to {\mathfrak U}_q(\mathfrak{gl}_{\cal N})\otimes {\mathfrak U}_q(\mathfrak{gl}_{\cal N})$. 
%Let $\Pi$ be 
%the `shift operator'   $~\Pi:\ {\cal U}_{1} \otimes {\cal U}_{2}\ \to\  {\cal U}_{2}  \otimes {\cal U}_{1}$   then  
%\begin{equation} 
%\Delta'(Y)\   = \Pi \circ \Delta(Y) ,~~~Y \in {\mathfrak U}_q(\mathfrak{gl}_{\cal N}). \label{perm} 
%\end{equation}
The $N$-fold co-product may be derived by using the recursion relations
\begin{equation}
 \Delta^{(N)} = (\mbox{id} \otimes \Delta^{(N-1)}) \Delta =  ( \Delta^{(N-1)}  \otimes \mbox{id})  \Delta, \label{cop2} 
\end{equation}
and as is customary, $\Delta^{(2)} = \Delta$ and  $\Delta^{(1)} = \mbox{id}$.

 Let us now consider the fundamental representation of ${\mathfrak U}_q(\mathfrak{gl}_{\cal N})$ \cite{Jimbo},
$\pi: {\mathfrak U}_q(\mathfrak{gl}_{\cal N})\to \mbox{End}({\mathbb C}^{\cal N})$:
\begin{eqnarray} 
\pi(e_{i})= e_{i, i+1}, ~~~~\pi(f_{i})=e_{i+1, i}, ~~~~\pi (q^{{\epsilon_{i} \over 2}}) = q^{{ e_{i,i} \over 2}}, \label{eval} 
\end{eqnarray}
and let us also introduce some useful notation:
\begin{equation}
(\pi \otimes \pi)\Delta(e_j) = \Delta(e_{j, j+1}), ~~~(\pi \otimes \pi)\Delta(f_j) = \Delta(e_{ j+1, j}), ~~~
(\pi \otimes \pi)\Delta(q^{\epsilon_j}) = \Delta(q^{e_{j, j}}). \label{notation}
\end{equation}

\begin{cor} \label{Prop4} The element ${\mathrm G}$ defined in (\ref{specialq}) is ${\mathfrak U}_q(\mathfrak{gl}_{\cal N})$ symmetric, i.e.
\begin{equation}
\Big [ {\mathrm G},\  \Delta_i(Y) \Big ] =0,  ~~~~Y \in\Big  \{ e_{j, j+1},\ e_{j+1,j},\ q^{e_{j,j}} \Big \} \label{symm1q}
\end{equation}
where we define the modified co-products ($i= 1, 2$): 
\begin{eqnarray}
&&\Delta_1(q^{e_{i,i}}) = q^{e_{\sigma(i),\sigma(i)}} \otimes q^{e_{i,i}}, ~~~\Delta_2(q^{e_{i,i}}) = q^{e_{i,i}} \otimes 
q^{e_{\tau(i), \tau(i)}}\nonumber\\ 
&&\Delta_1(\xi) =  \xi_{\sigma} \otimes q^{{H_j\over 2}}  + q^{-{H_{\sigma(j)}\over 2}}  \otimes \xi, \nonumber\\
&&\Delta_2(\xi ) = \xi \otimes   q^{{H_{\tau(j)}\over 2}} + q^{-{H_j\over 2}} \otimes  \xi_{\tau} . \label{symm2q}
\end{eqnarray}
 $H_j =\big (e_{j,j} -e_{j+1, j+1} \big )$, $H_{F(j)} = \big (e_{F(j), F(j)} -e_{F(j+1), F(j+1)}\big )$, for 
$\xi \in \Big  \{ e_{j, j+1},\ e_{j+1, j}\Big \}$,  
we define respectively: $\xi_{F} \in \Big \{ e_{F(j), F(j+1)},\  e_{F(j+1), F(j)}\Big \}$.
\end{cor}
\begin{proof}
This can be shown in a straightforward manner from the properties of (\ref{specialq}). 
Indeed, ${\mathrm g}$ (\ref{braidq})  is $\mathfrak{U}_q(\mathfrak{gl}_{\cal N})$ invariant \cite{Jimbo, Jimbo2}
(recall the fundamental representation (\ref{eval}))
\begin{equation}
\Big [ {\mathrm g},\  \Delta(Y) \Big ] =0,   \label{comm1q}
\end{equation}
wher  $Y \in\Big  \{ e_{j, j+1},\ e_{j+1,j},\ q^{e_{j,j}} \Big \}$ and the co-products of the algebra elements are given in (\ref{cop}) (see also (\ref{eval}), (\ref{notation})). We consider two invertible linear maps: ${\cal F}^{(2)}_i: \mbox{End}({\mathbb C}^{\cal N} \otimes {\mathbb C}^{\cal N}) \to  \mbox{End}({\mathbb C}^{\cal N} \otimes {\mathbb C}^{\cal N}), i \in \{ 1, 2\}$ such that ${\cal F}^{(2)}_1 := {\mathbb V} \otimes I$ and ${\cal F}^{(2)}_2 := I \otimes {\mathbb V}^{-1},$ where ${\mathbb V}$ is defined in Proposition \ref{Prop3}, then from (\ref{comm1q})
\begin{eqnarray}
{\cal F}^{(2)}_i \Big [ {\mathrm g},\  \Delta(Y) \Big ] ({\cal F}^{(2)}_i)^{-1}=0\  \Rightarrow\    \Big [ {\mathrm G},\  \Delta_i(Y) \Big ]  =0,  \label{comm2q}
\end{eqnarray}
where the modified co-produts are  defined as 
$\Delta_i(Y) ={\cal F}^{(2)}_i\Delta(Y) ({\cal F}^{(2)}_i)^{-1},$  and more specifically:
%Then (\ref{symm1q}) immediately follows from (\ref{comm1q}) and (\ref{specialq}) after 
%multiplying (\ref{comm1q}) from the left and right with ${\mathbb V}\otimes I,\  {\mathbb V}^{-1}  \otimes I$  or  
%$ I \otimes {\mathbb V}^{-1},\   I \otimes {\mathbb V}$ respectively.  
%The modified co-products of  ${\mathfrak U}_q(\mathfrak{gl}_{\cal N})$ 
%are then defined as:
\begin{eqnarray}
&& \Delta_1(q^{e_{i,i}}) =  {\mathbb V}q^{e_{i,i}} {\mathbb V}^{-1} \otimes q^{e_{i,i}}, ~~~~\Delta_2(q^{e_{i,i}}) = q^{e_{i,i}} \otimes
  {\mathbb V}^{-1} q^{e_{i,i}} {\mathbb V}, \nonumber\\
&&  \Delta_1(\xi) =  {\mathbb V} \xi  {\mathbb V}^{-1} \otimes q^{{H_j\over 2}}+   {\mathbb V}q^{-{H_j\over 2}}  {\mathbb V}^{-1}
\otimes \xi,  \nonumber\\ 
&&  \Delta_2(\xi) = \xi \otimes  {\mathbb V}^{-1}q^{{H_j\over 2}} {\mathbb V}+  q^{-{H_j\over 2}} \otimes  {\mathbb V}^{-1}   \xi 
 {\mathbb V},~~~~~\xi \in \Big \{e_{j,j+1},\ e_{j+1,j}\Big \}
\end{eqnarray}
and explicitly given by (\ref{symm2q}).

The coproducts $\Delta(Y)$  satisfy the ${\mathfrak U}_q(\mathfrak{gl}_{\cal N})$ relations, then via\\ $\Delta_i(Y) ={\cal F}^{(2)}_i\Delta(Y) ({\cal F}^{(2)}_i)^{-1},$  the modified coproducts $\Delta_i(Y)$ also satisfy the ${\mathfrak U}_q(\mathfrak{gl}_{\cal N})$ exchange  relations, thus 
${\mathrm G}$ is ${\mathfrak U}_q(\mathfrak{gl}_{\cal N})$ symmetric.
\end{proof}

 Explicit expressions for the modified $N$ co-products are then given as:
\begin{eqnarray}
&& \Delta^{(N)}_1(q^{e_{j,j}})=\bigotimes_{n=1}^N q^{e_{\sigma^{N-n}(j), \sigma^{N-n}(j)}}, ~~~ 
\Delta^{(N)}_2(q^{e_{j,j}})=\bigotimes_{n=1}^N q^{e_{\tau^{n-1}(j),\tau^{n-1}(j)}}\nonumber\\
&& \Delta_1^{(N)}(\xi) =\sum_{n=1}^N q^{-{H_{\sigma^{N-1}(j)}\over 2}}  \otimes \ldots  \otimes 
q^{-{H_{\sigma^{N-n+1}(j)}\over 2}}  \otimes \xi_{\sigma^{N-n}}
\otimes  q^{{H_{\sigma^{N-n-1}(j)}\over 2}}  \ldots \otimes  q^{{H_j\over 2}}, \nonumber\\
&& \Delta_2^{(N)}(\xi) =\sum_{n=1}^N  q^{-{H_j\over 2}}  \otimes \ldots  \otimes q^{-{H_{\tau^{n-2}(j)}\over 2}}  \otimes  \xi_{\tau^{n-1}}
\otimes  q^{{H_{\tau^{n}(j)}\over 2}}  \ldots \otimes  q^{{H_{\tau^{N-1}(j)}\over 2}}, \nonumber\\ 
&& \label{delta2q}
\end{eqnarray}
where  $\xi_{F^n} \in \Big \{ e_{F^n(j), F^n(j+1)},\  
 e_{F^n(j+1), F^n(j)}\Big \}$. The above expressions can be written in a compact form as:
$\Delta^{(N)}_i(Y) = {\cal F}_i^{(N)}\Delta^{(N)}(Y)({\cal F}_i^{(N)})^{-1},$ where  recall $Y \in\Big  \{ e_{j, j+1},\ e_{j+1,j},\ q^{e_{j,j}} \Big \},$ $\Delta^{(N)}(Y)$ are the ${\mathfrak U}_q(\mathfrak{gl}_{\cal N})$  $N$-coproducts,  and we define 
${\cal F}_{1}^{(N)}:= {\mathbb V}^{N-1} \otimes {\mathbb V}^{N-2} \otimes \ldots \otimes {\mathbb V} \otimes I$ and 
${\cal F}_2^{(N)} := I \otimes {\mathbb V}^{-1} \otimes {\mathbb V}^{-2}\otimes \ldots \otimes {\mathbb V}^{-(N-1)}$ (see also  \cite{DoikouDr}).

%%\noindent Note that Remark 2 of the previous subsection hods also true in the deformed case, 
%then the deformed co-products can be expressed as in (\ref{symm2q}), and
%{\em $H_j =\big (e_{j, j} -e_{j+1, j+1} \big )$, $H_{F(j)} = \big (e_{F(j),F(j)} -e_{F(j+1), F(j+1)}\big )$, and for 
%$\xi \in \Big  \{ e_{j, j+1},\ e_{j+1, j}\Big \}$,  
%we define respectively:\\ 
%$\xi_{\sigma} \in \Big \{ C_{\sigma(j)}\bar C_{\sigma(j+1)}e_{\sigma(j), \sigma(j+1)},\  
%C_{\sigma(j+1)}\bar C_{\sigma(j)}e_{\sigma(j+1), \sigma(j)}\Big \}$,\\
%$\xi_{\tau} \in \Big \{ \bar C_{j} C_{j+1} e_{\tau(j), \tau(j+1)},\  \bar C_{j+1} C_{j} e_{\tau(j+1), \tau(j)}\Big \}$.}
%3. Of course for the permutation operator itself $~V={\mathbb I}$.

$ $

Some general comments are in order here. We should note that set-theoretic solutions from 
braces have no semi-classical analogue \cite{DoiSmo}, thus they are fundamentally 
different from the known Yangian solutions or the $q$-deformed solutions of the YBE
associated to $\mathfrak {gl}_{\cal N}$ or $\mathfrak {U}_q(\mathfrak {gl}_{\cal N})$ \cite{Jimbo, ChaPre, Molev}.  
 This is evident even in the simple case of Lyubashenko's solution (please see  Proposition \ref{prop1} and simple examples 1 and 2 in page 9), recall 
$r = {\mathbb V}^{-1} \otimes {\mathbb V} \Rightarrow R(\lambda) =  \lambda {\mathbb V}^{-1} \otimes {\mathbb V}+{\cal P},$ where 
${\mathbb V} = \sum_{x,\in X} e_{x, \tau(x)}$ (more generally due to Proposition \ref{Drin}, $R(\lambda) = \lambda {\cal F}^{(op)} {\cal F}^{-1}+ {\cal P}$ and ${\cal F}^{(op)} = {\cal P} {\cal F} {\cal P}$). Such $R$-matrices can not be expressed as 
$1 + \hbar r^{(1)} + ...$ (up to an overall multiplicative function $f(\lambda)$),  given the form of ${\mathbb V}$  (or ${\cal F}$ explicitly given in \cite{DoikouDr}),
a fact that makes our construction 
distinct compared to the known examples of quantum algebras (quasi triangular Hopf algebras) 
as described for instance by Drinfeld in \cite{Drinfeld} (a detailed analysis on these issues is presented in [22]).
In this spirit it would be also very interesting to consider general twists, in analogy to Proposition \ref{Drin}, 
for the $q$-deformed case as well as the 
corresponding quantum groups and make possible connections with the theory of braces.

\section{Co-ideals: reflection $\&$ twisted algebras}

\noindent We introduce two, in principle distinct, quadratic algebras associated to the classification
of boundary conditions in  quantum integrable models. To define these quadratic algebras
in addition to the $R$-matrix we also need to introduce the $K$-matrix, 
which physically describes the interaction of particle-like excitations displayed by 
the quantum integrable  system, with the boundary of the system. 
The $K$-matrix satisfies \cite{Cherednik, Sklyanin, Olsha}:    
\begin{equation}
R_{12}(\lambda_1 - \lambda_2) {\mathbb K}_1(\lambda_1) \hat R_{12}(\lambda_1+\lambda_2) {\mathbb K}_{2}(\lambda_2) =  
{\mathbb K}_{2}(\lambda_2)  \hat R_{21}(\lambda_1+\lambda_2) {\mathbb K}_1(\lambda_1 )R_{21}(\lambda_1 - \lambda_2), \label{RE2}
\end{equation}
where we define in general $A_{21} = {\cal P}_{12} A_{12} {\cal P}_{12}$.
We make two distinct choices for $\hat R$, which lead to the two district quadratic algebras:
\begin{eqnarray}
&&\hat R_{12}(\lambda) = R_{12}^{-1}(-\lambda) ~~~~~~~~~~~~\mbox{Reflection algebra}\label{quadratic1} \\
&&\hat R_{12}(\lambda) = R^{t_1}_{12}(-\lambda -{{\cal N} \over 2}) ~~~~~~\mbox{Twisted algebra}, \label{quadratic2}
\end{eqnarray}
notice ${{\cal N}\over 2}$ is the Coxeter number for $\mathfrak{gl}_{\cal N}$.

In the self-conjugate cases e.g. in the case of e.g.
$\mathfrak {sl}_2,\ \mathfrak {U}_q(\mathfrak {sl}_2)$ or $\mathfrak {so}_n,\  \mathfrak {sp}_n $ $R$-matrices  
$R(\lambda) \sim {\cal C}_1 R_{12}^{t_1}(-\lambda -c){\cal C}_1$, for some matrix  ${\cal C}: {\cal C}^2 =I$ ,
i.e. the $R$-matrix is crossing symmetric,  and the two algebras, twisted and refection, coincide.
The constant  $c$ is associated to the Coxeter number of the corresponding algebra.
It is worth noting that these algebras are linked to two distinct types of integrable
boundary conditions, extensively studied in the context of  $A^{(1)}_{{\cal N} -1}$ affine Toda field theories \cite{Corrigan, DeMa, DoikouATFT}, 
and quantum spin chains \cite{Sklyanin} associated to 
$\mathfrak{gl}_{\cal N}$,  ${\mathfrak U}_q(\mathfrak{gl}_{\cal N})$,  and $\mathfrak{gl}({\cal N}|{\cal M})$ algebras \cite{deVega}, \cite{MeNeFu},
 \cite{DoikouNepomechie}-\cite{ DoikouMurphy}.

\subsection{Boundary Yang-Baxter equation $\&$ $B$-type Hecke algebra}

\noindent Let us first focus in the case where $\hat R_{12} (\lambda)= R_{12}^{-1}(-\lambda)\propto R_{21}(\lambda)$,  i.e. 
we consider the boundary Yang-Baxter or reflection equation \cite{Cherednik, Sklyanin}, expressed in the braid form
\begin{equation}
\check R_{12}(\lambda_1-\lambda_2){\mathbb K}_1(\lambda_1) \check R_{12}(\lambda_1 +\lambda_2) {\mathbb K}_1(\lambda_2)= 
{\mathbb K}_1 (\lambda_2) \check R_{12}(\lambda_1 +\lambda_2) {\mathbb K}_1(\lambda_1) \check R_{12}(\lambda_1-\lambda_2). \label{RE}
\end{equation}
As in the case of the Yang-Baxter equation, where representations of the $A$-type Hecke algebra are associated to solutions of the
Yang-Baxter equation \cite{Jimbo}, via the Baxterization process,  representations of the $B$-type Hecke algebra provide solutions 
of the reflection equation \cite{LevyMartin, DoikouMartin}.

\begin{defn} \label{def2}{\it The $B$-type Hecke algebra ${\cal B}_N(q, Q)$ is defined by the generators $g_l$, $l \in \{1,\ 2, \ldots, N-1  \}$ 
and $G_0$ and the exchange relations
(\ref{h1})-(\ref{h3}) and}
\begin{eqnarray}
&& G_0\ g_1\ G_0\ g_1 = g_1\ G_0\ g_1\ G_0 \label{b1}\\
&& \Big [ G_0,\ g_l\Big ] =0, ~~~l>1 \label{b2} \\
&&  \big ( G_0 - Q\big ) \big (G_0-Q^{-1} \big )= 0.\label{b3}
\end{eqnarray}
\end{defn}
We focus here on the case where $q=1$ and $Q$ arbitrary, and consider the brace solutions (\ref{brace1}) as representation of the Hecke elements $g_l$. 
We can solve the quadratic relation (\ref{b1}) together with (\ref{b3}) to provide representation of the $G_0$ element. Then via Baxterization 
we are able to identify suitable solutions of the reflection equation.
It is obvious that the identity is a solution of the relations (\ref{b1}), (\ref{b3}), and hence of the reflection equation.

\begin{rem} \label{Proposition 5} {\it  Let ${\mathrm b} = \sum_{x,z \in X} {\mathrm b}_{z,w}e_{z,w}$ be a representation of the $G_0$ 
element of the $B$-type Hecke algebra and $\check r$ is the set-theoretic solution given  in (\ref{brace1}).
Representations of $G_0$ can be identified.}

Indeed, let us  solve the quadratic relation (\ref{b1})
\begin{eqnarray}
&& ({\mathrm b} \otimes I)\  \check r\  ({\mathrm b} \otimes I)\ \check r= \check r\ ({\mathrm b} 
\otimes I)\  \check r\  ({\mathrm b} \otimes I).
\end{eqnarray}
The LHS of the latter equation leads to
\begin{equation}
\sum {\mathrm b}_{z,x} {\mathrm b}_{\sigma_x(y),\hat x} e_{z, \sigma_{\hat x}(\hat y)}\otimes e_{y, \tau_{\hat y}(\hat x)}, \label{L2}
\end{equation}
subject to: $\hat y = \tau_{y}(x)$,
whereas the RHS gives:
\begin{equation}
\sum {\mathrm b}_{\sigma_x(y), \hat x} {\mathrm b}_{\sigma_{\hat x} (\hat y), \hat w} e_{x,\hat w} 
\otimes e_{y, \tau_{\hat y}(\hat x)}\label{R2}
\end{equation}
subject to: $\hat y = \tau_{y}(x)$. Comparison of the LHS and RHS provide conditions among ${\mathrm b}_{x,w}$.
Moreover, ${\mathrm b}$ should satisfy condition (\ref{b3}) of the $B$-type Hecke algebra, which leads to
\begin{equation}
\sum_{y} {\mathrm b}_{z,y} {\mathrm b}_{y,w} = (Q -Q^{-1}) {\mathrm b}_{z,w} + \delta_{z,w}. \label{b4}
\end{equation}
Study of the fundamental relations above for any brace solution will lead to admissible representations for $G_0$.
\end{rem}

Note that in the special case that 
${\mathrm b}_{z,w} = \delta_{w, k(z)}$, where $k: X \to X$ satisfies $k(k(x))=x$ ($Q=1$), and some extra conditions 
that are discussed in the subsequent subsection, one recovers set-theoretic 
reflections (see also next subsection and  \cite{SmoVenWes} for a more detailed discussion).
In general, the full classification of  representations of the $B$-type Hecke algebra using the brace $\check r$-matrix (\ref{brace1}) 
is an important problem itself, which however will be left for future investigations.

\begin{rem} \label{rem4} { Let $\check r: V \otimes V \to V \otimes V,\ {\mathrm b}: V \to V$ provide  representations of the $B$-type Hecke algebra,
and assume that there exists some invertible ${\mathbb V}: V \to V$ (see also Lemma \ref{extra1} and Proposition \ref{prop1}):
\begin{equation}
({\mathbb V}\otimes{\mathbb V })\check r = \check r ({\mathbb V} \otimes {\mathbb V}).
\end{equation} 
We also define
\begin{equation}
\check \rho= ({\mathbb V } \otimes I)\ \check r \  ( {\mathbb V }^{-1}\otimes I)= (I \otimes {\mathbb V }^{-1})\ \check r\  (I\otimes {\mathbb V }),~~~~~
\beta = {\mathbb V } {\mathrm b} {\mathbb V }^{-1} \label{defGb}
\end{equation} 
It then follows that  $\check r,\ {\mathrm b}$ as well as $\check \rho,\ \beta$ provide  presentations of the $B$-type Hecke algebra.}
\end{rem}

\begin{rem} \label{rem5} {\it  Let ${\mathrm b}$ be an ${\cal N} \otimes {\cal N}$ matrix and $\check r$ be 
an ${\cal N}^{2} \otimes{\cal N}^2$ matrix.
Let also ${\mathrm b}_1$ (index notation) be a tensor realization of the $G_0$ element of 
the $B$-type Hecke algebra ${\cal B}_N(q=1, Q)$ and $\check r_{ll+1}$
a tensor realization of the element $g_l$ of ${\cal B}_N(q=1, Q)$. Then solutions of the reflection equation 
(\ref{RE}) ($\check R(\lambda) = \lambda \check r + I^{\otimes 2}$) can be expressed as, 
up to an overall function of $\lambda$, (Baxterization):
\begin{equation}
{\mathbb K}(\lambda) =   \lambda \big ( {\mathrm b} - {\kappa \over 2} I\big )+ {\hat c\over 2} I , \label{baxter}
\end{equation}
where $\hat c$ is an arbitrary constant, $\kappa = Q - Q^{-1}$ and $I$ the ${\cal N} \times {\cal N}$ identity matrix.}
\end{rem}

\noindent This has been done in \cite{LevyMartin, DoikouMartin}, but we briefly review the procedure here, in the special case $q=1$.
Indeed, recall $\check R$ is given by (\ref{braid1}) 
and let ${\mathbb K}(\lambda) = \xi(\lambda) I + \zeta(\lambda) {\mathrm b}$ where the functions $\xi(\lambda),\ \zeta(\lambda)$ 
will be identified. We substitute the expressions 
for $\check R$ and  $K(\lambda)$ in the reflection equation (\ref{RE}) and use repeatedly relations (\ref{b1}), (\ref{b2}), then after various 
terms cancellations the reflection equation (\ref{RE}) becomes:
\begin{eqnarray}
2\lambda_1 \xi_1 \zeta_2 -2 \lambda_2 \zeta_1\xi_2 + \kappa (\lambda_1 - \lambda_2) \zeta_1 \zeta_2 =  0 \label{basic2}
\end{eqnarray}
where we define: $\zeta_i = \zeta(\lambda_i),\ \xi_i = \xi(\lambda_i)$ and $\kappa = Q - Q^{-1}$.
We divide (\ref{basic2}) by $\zeta_1\zeta_2$ (provided that this is nonzero) and set $Q_i = {\xi_i \over \zeta_i}$:
\begin{eqnarray}
&& 2\lambda_1 Q_1 -2 \lambda_2 Q_2 + \kappa (\lambda_1 -\lambda_2) = 0\ \Rightarrow \ Q_i = {\hat c\over 2\lambda_i} - {\kappa \over 2},
\end{eqnarray}
and the latter implies: ${\xi(\lambda) \over \zeta(\lambda)} = {\hat c - \lambda \kappa \over 2 \lambda}$ ($\hat c$ is an arbitrary constant).

The remark above \ref{rem5}  is of course valid at the abstract level, 
that is solutions of the spectral dependent braid and reflections equations can
 be expressed in terms of the generators $g_l,\ G_0$ of the $B$-type Hecke algebra 
${\cal B}_N(q=1, Q)$, i.e. $\check R_{ll+1}(\lambda) = \lambda g_l + id$ and  
${\mathbb K}_1(\lambda) =   \lambda \big ( G_0 - {\kappa \over 2} id\big )+ {\hat c\over 2}id$.

\subsection{Set-theoretic representations of $B$-type Hecke algebras}

\noindent In this section we further investigate connections between the $B$-type Hecke algebra and 
the set-theoretic reflection equation, and give some specific  examples of
 representations of Hecke algebras that correspond to set-theoretic reflections.

\begin{lemma}\label{reflection}
 Let $(X, {\check r })$ be an involutive non-degenerate set-theoretic solution of the braid equation where 
${\check r}(x,y)=(\sigma_x(y), \tau_y(x))$. 
 Then $(X, {\check r }')$ is an involutive non-degenerate set-theoretic solution of the braid equation where 
${\check r}'(x,y)=(\tau_x(y), \sigma_y(x))$. 

 Let $k:X\rightarrow X$ be a function. 
Then the following are equivalent:
\begin{enumerate}

 \item  $k:X\rightarrow X$ is a solution to the set-theoretic reflection equation for the solution $(X,  {\check r})$: 
\[{\check r}K_{[1]}{\check r}K_{[1]}=K_{[1]}{\check  r}K_{[1]}{\check r}\] where $K_{[1]}(x,y)=(k(x), y)$.
 \item  $k:X\rightarrow X$ is a solution to the following  version of the reflection equation considered in \cite{SmoVenWes} 
for the solution $(X,  {\check r}')$: 
\[{\check r}'K_{[2]}{\check r}'K_{[2]}=K_{[2]}{\check  r}'K_{[2]}{\check r}'\] where $K_{[2]}(x,y)= (y, k(x))$.
\end{enumerate}
\end{lemma}
\begin{proof} Observe that $\check r$ is non-degenerate, hence maps 
 $\sigma_x, \tau_y$ are bijections. Consequently, $\check r$ is non-degenerate.
 Let $P:X\times X\rightarrow X\times X$ be defined as usually as $P(x,y)=(y, x)$ for $x,y\in X$.
Observe that ${\check r }'=P \check r P$, indeed $P{\check r}P(x,y)=P{\check r}(y,x)=P(\sigma _{y}(x), \tau _{x}(y))={\check r}'(x,y)$. 

Notice that ${\check r}'$ is involutive: ${\check r}'{\check r}'=P\check r PP\check r P=P\check r ^{2}P=P^2= id_{X\times X}$. Observe that 
\[{\check r}'K_{[2]}{\check r}'K_{[2]}=K_{[2]}{\check  r}'K_{[2]}{\check r}'\]  
 is equivalent to
\[(P{\check r}'P)(PK_{[2]}P)(P{\check r}'P)PK_{[2]}P=(PK_{[2]}P)(P{\check  r}'P)(PK_{[2]}P)(P{\check r}'P),\]
which immediately leads to
\[{\check r}K_{[1]}{\check r}K_{[1]}=K_{[1]}{\check  r}K_{[1]}{\check r}.\]

It remains to check that ${\check r}'$ is also a solution to the braid equation. For this purpose let us 
introduce, in the index notation, $P_{13}$: $P_{13}(x,y,z) = (z,y,x)$, it then follows that
$P_{13}(\check  r\times   id_X) P_{13} = id_X \times \check r'$ and $P_{13}( id_X \times  \check r) P_{13} = \check r' \times id_X$.
This is easily shown, indeed $P_{13}(\check  r\times  id_X) P_{13}(x,y,z) =P_{13}(\check  r\times  id_X)(z,y,x)  =
 P_{13} (\sigma_z(y), \tau_{y}(z), x) = (x, \tau_{y}(z) , \sigma_z(y))= ( id_X \times \check r') (x, y,z)$. 
Similarly, we show that
$P_{13}(  id_X\times  \check r) P_{13} = \check r' \times  id_X$.
By acting on the braid equation for $\check r$ with $P_{13}$ from the left and right  it then
immediately follows that $\check r'$ also satisfies the braid relation.
\end{proof}

Examples of functions $k$ satisfying the reflection equation related to  braces can be found in \cite{SmoVenWes, Katsa, DeCommer}.
Recall that this set-theoretical version of the reflection equation  together with the first examples of solutions first appeared in the work 
of Caudrelier and Zhang \cite{CauZha}

Notice that the element ${\mathbb b}$ of the Hecke algebra can be used to construct $c$-number $K$-matrices  satisfying equation (\ref{RE}), 
provided that $Q=1$.
Hence, by Lemma \ref{reflection}, constant  $K$-matrices  can be  obtained from involutive set-theoretic solutions to the reflection equation.
In particular, involutive $\tau $-equivariant functions give $c$-number solutions of the parameter dependent  equation (\ref{RE}), and every linear 
combination over $\mathbb C$ of such $K$-matrices  is also a constant  $K$-matrix, and hence gives a solution to  equation (\ref{RE})
(by Theorem 5.6 \cite{SmoVenWes} applied with interchanging $\sigma $ and $\tau $).

As an application of Lemma \ref{reflection} we obtain: 

\begin{pro}\label{Prop7} Let $(X,\check r)$ be an involutive, non-degenerate solution of the braid equation.  
Let  $\check r=\sum_{x,y\in X}e_{x, \sigma  _{x}(y)}\otimes e_{y, \tau _{x}(y)}$, and let
 $g_{n}= I^{\otimes n-1}\otimes {\check r} \otimes I^{\otimes N-n-1}$.
 Let ${\mathrm b} = \sum_{x\in X} e_{x,k(x)}$ for some function $k: X\rightarrow X$ such that $k(k(x))=x$ for all $x\in X$. 
Then ${\mathrm b}\otimes I$  is a representation  of the $G_0$ 
element of the $B$-type Hecke algebra (together with $\check r$ used for representation of elements $g_{n}$) if and only if 
\[ \tau_{\tau _{y}(x)}(k(\sigma _{x}(y))) =\tau_{\tau _{y}(k(x))}(k(\sigma _{k(x)}(y))).\]
\end{pro}
\begin{proof} This follows immediately from Lemma \ref{reflection} and Theorem 1.8  from \cite{SmoVenWes}, 
when we interchange $\sigma $ with $\tau $. 
\end{proof}
Let $(X,\check r)$ be an involutive, non-degenerate solution of the braid equation where we denote 
${\check r}(x,y)=(\sigma _{x}(y), \tau _{y}(x))$, and let $k:X\rightarrow X$ be a function. 
We say that $k$ is $\tau $-equivariant if for every $x, y\in X$ we have 
\[\tau _{x}(k(y))=k(\tau _{x}(y)).\]

It was shown in  \cite{SmoVenWes} that every function $k:X\rightarrow X$  satisfying $k(\sigma _{x}(y))=\sigma _{x}(k(y))$
 satisfies the set-theoretic reflection equation. By interchanging   $\sigma $ with $\tau $ and applying Lemma \ref{reflection} we get:
\begin{cor} Let $(X,\check r)$ be an involutive, non-degenerate solution of the braid equation.  
Let $\check r=\sum_{x,y\in X}e_{x, \sigma _{x}(y)}\otimes e_{y, \tau _{x}(y)}$, and let $g_{n}= 
I^{\otimes n-1}\otimes {\check r} \otimes I^{\otimes N-n-1}$.
 Let ${\mathrm b} = \sum_{x\in X} e_{x,k(x)}$ for some $\tau $-equivariant  function $k: X\rightarrow X$ such that $k(k(x))=x$ for all $x\in X$. 
Then ${\mathrm b}\otimes I$  is a representation of the  $G_0$ 
element of the $B$-type Hecke algebra (together with $\check r$ used for representation of elements $g_{n}$ in this Hecke algebra). 
\end{cor} Examples of $\tau $-equivariant functions can be defined  by fixing $x,y\in X$ and defining for
 $k(r)=\tau _{z}(y)$  for $r=\tau_{z}(x)$ (provided that  $\tau _{v}(x)=x$ implies $\tau_{v}(y)=y$ for every $v\in X$).
In \cite{Katsa} Kyriakos Katsamaktsis used central elements to construct $\mathcal {G}(X,r)$ equivariant functions, 
his ideas also allow to define 
$\tau $-equivariant functions in an analogous way- as $k(x)=\tau_{c}(x)$, where $c$ is central.

\subsection{Reflection $\&$ twisted algebras}

\noindent We shall discuss in more detail now the two distinct algebras associated to the quadratic equation (\ref{RE2}).
A solution of the quadratic equation (\ref{RE2}) is of the form \cite{Sklyanin, Olsha}
\begin{equation}
{\mathbb K}(\lambda |\theta_1)= L(\lambda -\theta_1)\  \big ( K(\lambda) \otimes \mbox{id}\big )\  \hat L(\lambda +\theta_1), \label{rep1}
\end{equation}
where $L(\lambda) \in\mbox{End}({\mathbb C}^{\cal N}) \otimes {\mathfrak A}$ satisfies the RTT  relation (\ref{RTT})
and $K(\lambda) \in \mbox{End}({\mathbb C}^{\cal N}) $ is a $c$-number solution of the quadratic equation (\ref{RE2})  (for some $R(\lambda)\in\mbox{End}({\mathbb C}^{\cal N} \otimes {\mathbb C}^{\cal N}),$ solution of the Yang-Baxter equation).  
We also define (in the index notation (see also Footnote 2, page 7))
\begin{eqnarray}
&& \hat L_{1n}(\lambda) = L_{1n}^{-1}(-\lambda) ~~~~~~~~~~~~\mbox{Reflection algebra}\nonumber\\
&& \hat L_{1n}(\lambda) = L^{t_1}_{1n}(-\lambda -{{\cal N} \over 2}) ~~~~~~\mbox{Twisted algebra}.
\end{eqnarray}
The quadratic algebra  ${\mathfrak B}$ defined by (\ref{RE2}) is a left co-ideal of the quantum algebra ${\mathfrak A}$ for
a given $R$-matrix (see also e.g. \cite{Sklyanin, DeMa, Doikou2}), i.e. the algebra  is endowed with a co-product 
$\Delta: {\mathfrak B} \to {\mathfrak B} \otimes {\mathfrak  A}$ \cite{Sklyanin}. Indeed, we define (in the index notation)
\begin{equation}
{\mathbb T}_{0;12}(\lambda|\theta_1,\theta_2) = L_{02}(\lambda-\theta_2) {\mathbb K}_{01}(\lambda|\theta_1)\hat L_{02}(\lambda+\theta_2), \label{rep3}
\end{equation}
where ${\mathbb K}(\lambda|\theta_1)$ is given in (\ref{rep1}) and in the index notation ${\mathbb K}_{01}(\lambda|\theta_1) = L_{01}(\lambda-\theta_1) K_0(\lambda) \hat L_{01}(\lambda+\theta_1)$.
Let also ${\mathbb K}_{01}(\lambda|\theta_1) = \sum_{a,b =1}^{\cal N}e_{a,b}\otimes {\mathbb K}_{a,b}(\lambda|\theta_1)\otimes\mbox{id} $,  $L_{02} = \sum_{a,b=1}^{\cal N}e_{a,b} \otimes \mbox{id} \otimes L_{a,b}(\lambda)$ and
 ${\mathbb T}_{0;12}(\lambda|\theta_1,\theta_2)= \sum_{a,b=1}^{\cal N}e_{a,b} \otimes \Delta({\mathbb K}_{a, b}(\lambda|\theta_1,\theta_2)),$ then via expression (\ref{rep3}):
\begin{equation}
\Delta({\mathbb K}_{a,b}(\lambda|\theta_1, \theta_2) )= \sum_{k,l} {\mathbb K}_{k,l}(\lambda|\theta_1) \otimes L_{a,k}(\lambda -\theta_2) \hat L_{l,b}(\lambda+\theta_2), \label{rep2}
\end{equation}
where the elements ${\mathbb K}_{k,l}(\lambda|\theta_1)$  can be also re-expressed in terms of the elements of the $c$-number matrix $K$ and $L$  when considering the realization (\ref{rep1}).
%It is useful for our purposes here to
%introduce the ``opposite'' co-product $\Delta' = \pi \circ \Delta$ where $\pi: a \otimes b \to b\otimes a$.
%The $n$ co-product is obtained by iteration:\\ $\Delta^{(n)} = \big (\mbox{id} \otimes \Delta^{(n-1)}\big ) \Delta$ and 
%$\Delta^{'(n)} = \big (\mbox{id} \otimes \Delta^{(n-1)}\big ) \Delta'$.  

In our analysis in the subsequent section, we shall be primarily focusing on tensor representations of ${\mathbb K}$ and on the special case: 
$L(\lambda) \to R(\lambda)$, $\ \hat L(\lambda) \to \hat R(\lambda)$ and for the rest of the present subsection and subsections 5.1-5.3 we shall  be considering  $R(\lambda) = \lambda {\cal P} \check r + {\cal P}$, 
where $\check r$ provides a representation of the $A$-type Hecke algebra ${\cal H}_N(q=1)$ and ${\cal P}$ is the permutation operator.
%In the case we are considering here
%\begin{equation}
%R(\lambda) = \lambda r + {\cal P}, ~~~~~\bar R (\lambda) = \check R(\lambda) {\cal P}= \lambda \bar r +{\cal P}. \label{brace3}
%\end{equation}
%$\bar r = \check r {\cal P} = \sum_{x,y}e_{x\tau_{y}(x)} \otimes e_{y\sigma_x(y)}$, and recall $r$ is given by (\ref{rr1}).

Before we move on with stating the next Proposition and Corollaries regarding the quadratic algebras defined by (\ref{RE2}) 
we first introduce some useful notation associated to both the reflection and twisted algebras (\ref{RE2}).
 We introduce $\check r^*$ and $\hat{\cal P}$:
\begin{eqnarray}
&& \check r^*_{12} =\check r_{12}, ~~~~~~\hat {\cal P}_{12} =I^{\otimes 2} ~~~~~~~~~\mbox{Reflection algebra} \label{qq1} \\
&&\check r_{12}^* =r_{12}^{t_1}{\cal P}_{12} ,\ ~~~\hat {\cal P}_{12} =\big  ({\cal N \over 2}r_{12}^{t_1} -{\cal P}_{12}^{t_1}\big ){\cal P}_{12}~~~~~~~\mbox{Twisted algebra}. \label{qq2}
\end{eqnarray}

\begin{pro} \label{defin} Let $\check R(\lambda)= \lambda\check r + I^{\otimes 2}$, where $\check r$ 
provides a tensor realization of 
the Hecke algebra ${\cal H}_N(q=1)$, and let ${\mathbb K}(\lambda)$ satisfy the quadratic equation (\ref{RE2}).
Let also ${\mathbb K}(\lambda) = 
\sum_{n=0}^{\infty}{{\mathbb K}^{(n)} \over \lambda^n}$ and $~{\mathbb K}^{(n)} = \sum_{z, w\in X} e_{z,w}\otimes {\mathbb K}_{z,w}^{(n)},$ where ${\mathbb K}_{z,w}^{(n)}$ are the generators of the quadratic algebra defined by (\ref{RE2}). 
The exchange relations among the
quadratic algebra generators are encoded in:
\begin{eqnarray}
& & \check r_{12}{\mathbb K}_1^{(n+2)} \check r^*_{12} {\mathbb K}_1^{(m)} - \check r_{12} {\mathbb K}_1^{(n)} 
\check r^*_{12} {\mathbb K}_1^{(m+2)} + 
\check r_{12} {\mathbb K}_1^{(n+1)} \hat {\cal P}_{12} {\mathbb K}_1^{(m)}\nonumber\\ && -\ \check r_{12} 
{\mathbb K}_1^{(n)} \hat {\cal P}_{12} {\mathbb K}_1^{(m+1)}
+  {\mathbb K}_1^{(n+1) }\check r^*_{12} {\mathbb K}_1^{(m)} +  {\mathbb K}_1^{(n) }\check r^{*}_{12} 
{\mathbb K}_1^{(m+1)}  +
{\mathbb K}_1^{(n)} \hat {\mathcal P}_{12} {\mathbb K}_1^{(m)}
\nonumber\\
&= & {\mathbb K}_1^{(m)} \check r^*_{12} {\mathbb K}_1^{(n+2)}\check r_{12} - {\mathbb K}_1^{(m+2)} 
\check r^*_{12} {\mathbb K}_1^{(n)}\check r_{12} + 
{\mathbb K}_1^{(m)} \hat {\cal P}_{12} {\mathbb K}_1^{(n+1)}\check r_{12} \nonumber\\ 
&& -\  {\mathbb K}_1^{(m+1)}\hat {\cal P}_{12}  {\mathbb K}_1^{(n)}\check r_{12}+  {\mathbb K}_1^{(m+1)} 
\check r^*_{12} {\mathbb K}_1^{(n)}+  
{\mathbb K}_1^{(m)} \check r^*_{12} {\mathbb K}_1^{(n+1)} + {\mathbb K}_1^{(m)}\hat 
{\mathcal P}_{12} {\mathbb K}_1^{(n)}, \nonumber \\
&& \label{RABasic}
\end{eqnarray}
where $\check r^*$ and $\hat {\cal P}$ are defined in (\ref{qq1}), (\ref{qq2}).
\end{pro}
\begin{proof}
First we act from the left and right of (\ref{RE2}) with the permutation operator ${\cal P},$ then (\ref{RE2}) becomes
\begin{equation}
\check R_{12}(\lambda_1-\lambda_2) {\mathbb K}_1(\lambda_1) \check R^*_{12}(\lambda_1+\lambda_2)
{\mathbb K}_1(\lambda_2)={\mathbb K}_1(\lambda_2)  \check R^*_{12}(\lambda_1+\lambda_2){\mathbb K}_1(\lambda_1)\check R_{12}(\lambda_1-\lambda_2), \label{RE3}
\end{equation}
where $\check R(\lambda_1 - \lambda_2) = (\lambda_1 - \lambda_2)\check r + I^{\otimes 2 }$ and $\check R^*(\lambda_1 +\lambda_2)= (\lambda_1 +\lambda_2) \check r^* + \hat {\cal P}$ ($\check r^*,\  \hat{\cal P}$ are defined in (\ref{qq1}), (\ref{qq2})), and we recall that ${\mathbb K}(\lambda_i) = \sum_{n=0}^{\infty}{{\mathbb K}^{(n)} \over \lambda_i^n}$ ($i\in \{1,\ 2\}$).
We substitute the above expressions in  (\ref{RE3}),  and we gather terms proportional to $\lambda_1^{-n} \lambda_2^{-m}$, $n, m\geq 0$ in the LHS and RHS of (\ref{RE3}), which lead to  (\ref{RABasic}). Recalling also that in general $A_{12} = A \otimes \mbox{id}_{\mathfrak A}$, 
$~{\mathbb K}^{(n)}_1 = \sum_{z, w\in X} e_{z,w} \otimes I\otimes {\mathbb K}_{z,w}^{(n)},$ 
and substituting the latter  expressions in (\ref{RABasic}) we obtain the exchange relations among the 
generators ${\mathbb K}_{z,w}^{(n)}$, which are particularly involved and are omitted here.
\end{proof}

It is useful for the following Corollaries to focus on terms proportional to $\lambda_1^2 \lambda_2^{-m}$ and $\lambda_1 \lambda_2^{-m}$ (or equivalently  $\lambda_2^2 \lambda_1^{-m}$ and $\lambda_2 \lambda_1^{-m}$) in the $\lambda_{1,2}$ expansion of the quadratic algebra, and obtain
\begin{eqnarray}
&& \check r_{12} {\mathbb K}_1^{(0)}\check r_{12}^* {\mathbb K}_1^{(m)}  =  
{\mathbb K}_1^{(m)} \check r_{12}^*{\mathbb K}_1^{(0)}\check  r_{12}  \label{RABasic2} \\
&& \check r_{12} {\mathbb K}_1^{(1)}\check r_{12}^* {\mathbb K}_1^{(m)} + {\mathbb K}_1^{(0)} \check r_{12}^*   {\mathbb K}_1^{(m)}
+  \check r_{12} {\mathbb K}_1^{(0)}\hat {\cal P}_{12} {\mathbb K}_1^{(m)}
=  \label{RABasic2b} \\
&& {\mathbb K}_1^{(m)} \check r_{12}^*{\mathbb K}_1^{(1)}\check  r_{12}    +  {\mathbb K}_1^{(m)} \check r_{12}^*  {\mathbb K}_1^{(0)}
+  {\mathbb K}_1^{(m)} \hat {\cal P}_{12} {\mathbb K}_1^{(1)}\check  r_{12}.\nonumber
\end{eqnarray}
The two Corollaries that follow concern the reflection algebra only, i.e. $\check r^* = \check r,\  \hat {\cal P} = I^{\otimes 2}$.

\begin{cor} \label{Proposition 6}{\it A finite non-abelian sub-algebra of the reflection algebra  exists, 
realized by the elements of ${\mathbb K}^{(1)}$ when ${\mathbb K}^{(0)} \propto I$.}
\end{cor}
\begin{proof}
We focus on terms proportional $\lambda_1^2 \lambda_2^{-m}$ and $\lambda_1 \lambda_2^{-m}$ (\ref{RABasic2}),  
(\ref{RABasic2b}) in the case of the reflection algebra:
\begin{eqnarray}
&& \Big [\check r_{12} {\mathbb K}_1^{(0)} \check r_{12},\ {\mathbb K}_1^{(m)} \Big ] =0 \label{rela1}  \\
&& \Big [ \check r_{12} {\mathbb K}_1^{(1)}\check r_{12},\ {\mathbb K}_1^{(m)} \Big ]= \label{rela2}\\
&&  \ {\mathbb K}_1^{(m)}  {\mathbb K}_1^{(0)} \check r_{12}  + {\mathbb K}_1^{(m)}  
\check r_{12}  {\mathbb K}_1^{(0)} - {\mathbb K}_1^{(0)} \check r_{12} {\mathbb K}_1^{(m)}
-  \check r_{12} {\mathbb K}_1^{(0)}{\mathbb K}_1^{(m)}.  \nonumber
\end{eqnarray}
Notice that due to (\ref{rep1})  in the case of the reflection algebra ${\mathbb K}^{(0)} 
\propto  I$ when the $c$-number matrix $K  \propto I$.
For $m=1$ equation (\ref{rela2}) provides the defining relations of a finite sub-algebra of the reflection algebra generated 
by ${\mathbb K}^{(1)}_{x,y}$ .
\end{proof}

%\noindent {\bf WE CAN ADD PROP ON EXCHANGE RELATIONS OF THE FINITE SUBALGERBA FOR THE BRACE SOLUTION.....}

\begin{cor} \label{Proposition 7} {\it For the special class of Lyubashenkos's solutions $\check r$ of Proposition \ref{prop1} a finite non-abelian sub-algebra
of the reflection algebra exists,  realized by the elements of ${\mathbb K}^{(1)}$ for any ${\mathbb K}^{(0)}$.
When ${\mathbb K}^{(0)} \propto I$ the finite sub-algebra generated by the ${\mathbb K}_{x,y}^{(1)}$ is the $\mathfrak{gl}_{\cal N}$ algebra.
Moreover, traces of ${\mathbb K}^{(m)}$ commute with the elements ${\mathbb K}_{x,y}^{(1)}$,}
\begin{equation}
\Big [ {\mathbb K}^{(1)}_{x,y},\ tr_1({\mathbb K}_1^{(m)}) \Big ] =0, ~~~~\forall x,\ y \in X. \label{commut}
\end{equation}
\end{cor}
\begin{proof}
\noindent For the special class of solutions $\check r_{12} = {\mathbb V}_1 {\cal P}_{12} {\mathbb V}_1^{-1}$  (\ref{special1}), equation (\ref{rela1}) becomes 
$\Big [ \tilde {\mathbb K}_2^{(0)},\  \tilde{\mathbb K}_1^{(m)}\Big ] =0$, where we define $\tilde {\mathbb K}^{(m)} = {\mathbb V}^{-1}{\mathbb K}^{(m)} {\mathbb V}$ (${\mathbb V} = \sum_{x\in X} e_{x, \tau(x)}$), which  reads for the matrix elements as:
${\mathbb K}^{(m)}_{x,y} = 
\tilde {\mathbb K}^{(m)}_{\tau(x), \tau(y)}$.  The latter commutator implies that ${\mathbb K}^{(0)}$ is a $c$-number matrix 
(i.e. the entries of ${\mathbb K}^{(0)}$  are $c$-numbers).
Also, (\ref{rela2}) becomes
\begin{eqnarray}
&& \Big  [ \tilde  {\mathbb K}_2^{(1)},\  \tilde {\mathbb K}_1^{(m)}\Big ] = 
{\cal P}_{12}  \Big (\tilde {\mathbb K}_2^{(m)} \big  (\tilde {\mathbb K}_1^{(0)}+ \tilde {\mathbb K}_2^{(0)}  \big ) -   
\big  (\tilde {\mathbb K}_1^{(0)}+ \tilde {\mathbb K}_2^{(0)} \big  )   
 \tilde {\mathbb K}_1^{(m)}\Big ).
 \label{suba2} 
\end{eqnarray}
Given that ${\mathbb K}^{(0)}$ is a $c$-number 
matrix we conclude that expression (\ref{suba2}) for $m=1$  provides a closed algebra formed by the elements of ${\mathbb K}^{(1)}$. For $m=1$ and for ${\mathbb K}^{(0)} \propto  I$ (\ref{suba2}) gives the $\mathfrak{gl}_{\cal N}$ exchange relations (up to an overall multiplicative factor, 
which can be absorbed by rescalling the  generators). See also relevant results on tensor realizations of the sub-algebra in Corollary 
\ref{Proposition 11}.

Taking the trace of (\ref{suba2}) with respect to space 1 and using 
$\big [ \tilde {\mathbb K}_2^{(0)},\  \tilde{\mathbb K}_1^{(m)}\big ] =0$ 
we arrive at (\ref{commut}). 
\end{proof}

\section{Open quantum spin chains $\&$ associated symmetries}

\noindent  We consider in what follows  spin-chain like representations, 
i.e. we are focusing on tensor representations of the quadratic algebra (\ref{RE2}) (see also (\ref{rep1})): 
$L(\lambda) \to R(\lambda)$, $\ \hat L(\lambda) \to \hat R(\lambda)$ and ${\mathbb K}_{01} (\lambda|\theta_1) \to {\cal K}_{01}(\lambda|\theta_1) = R_{01}(\lambda-\theta_1)  K_0(\lambda) \hat R_{01}(\lambda+\theta_1),$ where recall $K(\lambda)$ is a $c$-number solution of the quadratic equation (\ref{RE2}),  $R(\lambda)$ is a solution of the Yang-Baxter equation and $\hat R(\lambda)$  is defined in (\ref{quadratic1}), (\ref{quadratic2}).  

We introduce  the open monodromy matrix ${\cal T}_{0;12...N}(\lambda|\{\theta_i\}) \in \mbox{End}(({\mathbb C}^{\cal N})^{\otimes (N+1)})$ \cite{Sklyanin}, which provides a tensor representation of (\ref{RE2}):
\begin{equation}
{\cal T}_{0;12...N}(\lambda|\{\theta_i\})=  T_{0;12...N}(\lambda|\{\theta_i\})\ K_0(\lambda)\ \hat T_{0;12...N}(\lambda|\{\theta_i\}), \label{modif}
\end{equation}
where $\{\theta_i\} := \{\theta_1, \ldots, \theta_N\}$ and the monodromy matrix\\ $T_{0;12...N}(\lambda|\{\theta_i\}) \in \mbox{End}(({\mathbb C}^{\cal N})^{\otimes (N+1)})$ is given by 
\begin{equation}
T_{0;12...N}(\lambda|\{\theta_i\}) = R_{0N}(\lambda-\theta_N) \cdots R_{02}(\lambda-\theta_2) R_{01}(\lambda-\theta_1) \label{monod}
\end{equation}
and satisfies (\ref{RTT}). Also, $\hat T_{0;12...N}(\lambda|\{\theta_i\}) = T_{0;12...N}^{-1}(-\lambda|\{\theta_i\})$ 
in the case of the reflection algebra
and $\hat T_{0;12...N}(\lambda|\{\theta_i\}) = T_{0;12...N}^{t_0}(-\lambda -{{\cal N} \over 2}|\{\theta_i\})$ in the case of twisted algebra. 
We shall consider henceforth in expression (\ref{monod}) $\theta_i =0,\ i \in \{1, \ldots, N\}$. 
Such a choice is justified by the fact that we wish to construct local Hamiltonians, based on the fact that 
$R(0) \propto {\cal P}$ (${\cal P}$ the permutation operator), as will be transparent in the next subsection. The fact that the monodromy  matrix $T$ satisfies the RTT relation and $K$ is a $c$-number solution of 
the refection equation guarantee that the modified monodromy ${\cal T}$ also 
satisfies the reflection equation, 
The elements of the  modified monodromy matrix are ${\cal T}_{x,y}(\lambda) =
\Delta^{(N)}({\cal K}_{x,y}(\lambda))$ (see also discussion in the first paragraph of subsection 4.3).  
We also define the open or double row transfer matrix \cite{Sklyanin} as
\begin{equation}
\mathfrak{t}(\lambda) = tr_0\big (\hat K_0 {\cal T}_0(\lambda) \big), \label{transfer}
\end{equation}
where $\hat K$ is a solution of a dual quadratic equation\footnote{The dual quadratic equation is similar to (\ref{RE2}), 
but $\lambda_i \to -\lambda_i - {{\cal N} \over 2}$ in the arguments of $R,\hat R$.} (\ref{RE2}).
%and also via (\ref{RE2})it immediately follows that:
%\begin{equation}\Delta^{'(N+1)}({\mathbb K}(\lambda' -\lambda))\ {\cal T}(\lambda)\end{equation}where 
Note that for historical reasons
the space indexed by $0$ is usually called the {\it auxiliary space}, whereas the spaces indexed by $1, 2, \ldots, N$ are called {\it quantum spaces}.
Notice also that the  quantum indices are suppressed in the definitions of  $T,\ \hat T$ and   ${\cal T}$ for brevity.

To prove integrability of the open spin chain,  constructed from the brace $R$-matrix and the corresponding $K$-matrices
we make use of the two  important properties  for the $R$-matrix, i.e. the unitarity and crossing-unitarity (\ref{u1}) and (\ref{u2}) respectively.
Indeed, using the fact that ${\cal T}$ and $\hat K$ satisfy the quadratic and dual equations (\ref{RE2}), and also $R$ satisfies
the fundamental properties (\ref{u2}), (\ref{tt}) it can be shown that (see \cite{Sklyanin, Doikou1} 
for detailed proofs on the commutativity of the open transfer matrices associated to both reflection and twisted algebras):
\begin{equation}
\Big [{\mathfrak t}(\lambda),\  {\mathfrak t}(\mu) \Big ] =0. \label{invo1}
\end{equation}

We focus henceforth on the reflection algebra only, and we  investigate the symmetries associated to the 
open transfer matrix for generic boundary conditions. The main goal in the context of quantum integrable systems
is the derivation of the eigenvalues and eigenstates of the
transfer matrix. This is in general an intricate task and the typical methodology used is the Bethe ansatz
formulation, or suitable generalizations \cite{Korepin, FadTak}.
 In the algebraic Bethe anastz scheme the symmetries of the transfer matrices and the existence of a reference state are essential
components. 
When an obvious reference state is not available, which is the typical scenario when considering set-theoretic solutions, 
certain Bethe ansatz generalizations  can be used.  Specifically, the methodology implemented by Faddeev and Takhtajan 
in \cite{FadTak} to solve the XYZ model, based on the application of local gauge (Darboux)
transformations at each site of the spin chain can be used.  The Separation of Variables technique, introduced by Sklyanin \cite{Sklyanin2}, 
and recently further developed for open quantum spin chains \cite{Kitanine},  can also be employed, in particular when addressing 
the issue of Bethe ansatz completeness, but also as a further consistency check. Moreover, we plan to generalize the findings of \cite{Maillet} on
the role of Drinfeld twists in the algebraic Bethe ansatz, for set-theoretic solutions. This will lead to new
significant connections, for instance with generalized Gaudin-type models.

\subsection{Symmetries of the open transfer matrix}
\noindent We shall prove in what follows some fundamental Propositions that will provide significant information on the 
symmetries of the double row transfer matrix (\ref{transfer}). Note that henceforth we consider $\hat K \propto I$
 in (\ref{transfer}).

Let us first prove a useful lemma for the brace $\check r$ matrix.
\begin{lemma} \label{trace0} 
Let $(X, \check {r})$ be a finite, involutive, non-degenerate set-theoretic solution of the 
Yang-Baxter equation (i.e. a solution obtained from a finite brace). 
Let $\check r$ be the brace matrix $\check r = \sum_{x, y\in X} e_{x, \sigma_x(y)} 
\otimes e_{y, \tau_y(x)}$, then $tr_0( \check r_{n0}) = I$.
\end{lemma}
\begin{proof} Let $(X, {\check r})$ be our underlying set-theoretic solution.  
Recall that\\ $\check r=\sum_{x,y\in X}e_{x, \sigma _{x}(y)}\otimes e_{y, \tau_{y}(x)}.$ 
Observe that \[tr  _{0}(\check r _{n0})=\sum_{(x,y)\in W} e_{x, \sigma _{x}(y)},\] where $(x,y)\in W$  if and only if  $y=\tau_{y}(x)$. 
 Notice that if $(x,y)\in W$ then $x=\tau _{y}^{-1}(y)$.
Observe that  $\tau _{y}^{-1}(y)$  is always in the set $X$ (because our sets are finite so the inverse of map $\tau $ is some power of map $\tau $), 
so for each $y$ there exist $x$ such that $(x,y)\in W$.
This implies that  that for each $y$ in $X$ there is exactly one $x$ in $X$ such that $(x,y)$ is in $W$, we will  denote this $x$ as $x_{[y]}$.
 This  implies that $tr  _{0}(\check r_{n0})=\sum_{y\in X} e_{x_{[y]}, \sigma _{x_{[y]}}(y)}.$
 We notice that $\sigma _{x_{[y]}}(y)=x_{[y]}$,  it follows from the fact that  $(x_{[y]}, y)$ is in $W$.
Consequently, $tr  _{0}(\check r_{n0})=\sum_{y\in X} e_{x_{[y]}, x_{[y]}}.$
 We notice further that if $(x,y)$ in $W$  and $(x,z)$ in $W$ then $y=z$, so for each
$ x$ there is exactly one $y$ such that $(x,y)$ is in $W$.
Therefore,  \[tr  _{0}(\check r _{n0})=\sum_{z\in X} e_{z, z}\]  (where $z$ equals  elements $x_{[y]}$ for different $y$).
Hence, that  $tr  _{0}(\check r _{n0})=I$ where recall $I$ is the identity matrix of dimension equal to the cardinality of $X$.
\end{proof}

The following Proposition is quite general and holds 
for any $R(\lambda) =  \lambda {\mathcal P}\check r +{\mathcal P}$,  and 
$K(\lambda) = \lambda c ({\mathrm b} - {\kappa \over 2}I) + I$, ($c$ is an arbitrary constant and 
$\kappa =  Q- Q^{-1}$, see also Remark \ref{rem5}).
Also, $\check r$ and ${\mathrm b}$  provide a representation of the $B$-type 
Hecke algebra ${\cal B}_N(q=1, Q)$, and ${\mathcal P}$ is the permutation operator.
Recall we consider $\hat K =I$ in the definition of the open transfer matrix (\ref{transfer}).
 
\begin{pro}\label{Traces} Let $R(\lambda) =  \lambda {\mathcal P}\check r +{\mathcal P}$,  and 
$K(\lambda) = \lambda c ({\mathrm b} - {\kappa \over 2}I) + I,$ where $\check r$ and ${\mathrm b}$ provide representations of the the B-type Hecke algebra ${\cal B}_N(q=1, Q)$  and ${\cal P}$ is the permutation operator ($c$ is an arbitrary constant and 
$\kappa =  Q- Q^{-1}$).  Consider the $\lambda$-series expansion of the corresponding modified monodromy matrix (\ref{modif})  : ${\cal T}(\lambda) = 
\lambda^{2N+1}\sum_{k=0}^{2N+1}{{\mathcal T}^{(k)} \over \lambda^k}$, and the series expansion of the double row transfer matrix 
${\mathfrak t}(\lambda) = \lambda^{2N+1} \sum_{k=0}^{2N+1}{ {\mathfrak t}^{(k)} \over \lambda^k}$, where
${\mathfrak t}^{(k)} = tr_0 ({\mathcal T}_0^{(k)})$. 
Then the commuting quantities, ${\mathfrak t}^{(k)}$ for $k = 1, \ldots, 2N+1$, are expressed exclusively in terms 
of the elements $\check r_{n n+1}$, $n = 1, \ldots, N-1$, and ${\mathrm b}_1$, provided that $tr_0(\check r_{N0}) = I$.
\end{pro}
\begin{proof}
Let $T(\lambda) = \lambda^N \sum_{k=0}^N {T^{(k)} \over \lambda^k}$, $k \in \{0,1,\ldots, N \}$.
Let us also introduce some useful notation:
\begin{equation}
{\mathfrak T}^{(N-k-1)} = \sum_{[n_k, n_1]} \prod_{1\leq j \leq k}^{\leftarrow}\check  r_{n_j  n_{j}+1}, \qquad 
\hat {\mathfrak T}^{(N-k-1)} = \sum_{[n_k, n_1]} \prod_{1\leq j \leq k}^{\rightarrow} \check  r_{n_j  n_{j}+1}, \nonumber
\end{equation}
where we define $[n_k, n_1]: 1\leq n_k < ...<n_1\leq N-1$, and the ordered products are given as 
$\prod_{1\leq j\leq k}^{\rightarrow} \check r_{n_j  n_{j}+1 }  =\check  r_{n_k n_k+1}  
\check r_{n_2 n_2+1} \ldots \check r_{n_1 n_1+1}$, $\ \prod_{1\leq j\leq k}^{\leftarrow} \check r_{n_j  n_{j}+1 }  =\check  r_{n_1 n_1+1}  
\check r_{n_2 n_2+1} \ldots \check r_{n_k n_k+1}$, $\ n_1>n_2>\ldots n_k$.

In the proof of Proposition 4.1 in \cite{DoiSmo} all the members of the expansion of the monodromy $T^{(k)}$, were computed 
using the notation introduced above and the definition of the monodromy, and were expressed as:
$T_0^{(N-k)} = \Big ( {\mathfrak T}^{(N-k-1)} +  \check r_{N0}   {\mathfrak T}^{(N-k)}\Big ) {\mathcal P}_{01} \Pi$,
and similarly: $\hat T_0^{(N-k)} = \hat \Pi  {\mathcal P}_{01}\Big ( \hat {\mathfrak T}^{(N-k-1)} +    \hat {\mathfrak T}^{(N-k)} \check r_{N0} \Big ) $,
where $\Pi = {\cal P}_{12} \ldots {\cal P}_{N-1 N}$ and $\hat \Pi= {\cal P}_{N-1N} {\cal P}_{N-2 N-1} \ldots {\cal P}_{12}$.

Let us  also express the $c$-number $K$-matrix (\ref{baxter}) (derived up to an overall constant) as:
$K(\lambda) = \lambda \hat {\mathrm b} + I$,
where  $\hat {\mathrm b} =c \big ({\mathrm  b} -{\kappa \over 2}I\big )$ (see also (\ref{baxter})),
and recall here $\hat K = I$.
Also, in accordance to the expansion of the monodromy matrix in the previous section we express the modified monodromy 
as a formal series expansion: ${\cal T}(\lambda) =\lambda^{2N +1}\sum_k {{\cal T}^{(k)} \over \lambda^k} $,
then each term of the expansion is expressed as:
\begin{equation}
{\cal T}_0^{(2N-n+1)} = \sum_{k, l} T_0^{(N-k)} \hat {\mathrm b}_0 \hat T_0^{(N-l)}|_{k+l = n-1} + 
\sum_{k, l}T_0^{(N-k)}  \hat T_0^{(N-l)}|_{k+l = n}. \label{CT}
\end{equation}
After taking the trace and using the fact the $tr_0 (\check r_{N0}) = I$ we conclude for the first term of the expression (\ref{CT}) above:
\begin{eqnarray}
 tr_0\big(T_0^{(N-k)} \hat {\mathrm b}_0 \hat T_0^{(N-l)}\big)_{k+l = n-1}&=& 
{\mathfrak T}^{(N-k)} \hat {\mathrm b}_1   \hat {\mathfrak T}^{(N-l-1)}  +  {\mathfrak T}^{(N-k-1)} 
\hat {\mathrm b}_1   \hat {\mathfrak T}^{(N-l)}  \nonumber\\ 
&  &+  {\cal N} {\mathfrak T}^{(N-k-1)} \hat {\mathrm b}_1 \hat {\mathfrak T}^{(N-l-1)} \nonumber\\ & &+ 
tr_0 \big ( \check r_{N0} {\mathfrak T}^{(N-k)} \hat {\mathrm b}_1 \hat {\mathfrak T}^{(N-l)} \check r_{N0}\big ). \label{CT2}
\end{eqnarray}
Analogous expression is derived for the second term in (\ref{CT}), given that $\hat {\mathrm b} \to I$ and $k+l =n$
in the expression above. The first three terms of  (\ref{CT2}) 
are clearly expressed only in terms of the elements of the $B$-type Hecke algebra $\check r_{n n+1},\ {\mathrm b}_1$ (recall  
$\hat {\mathrm b} =c \big ({\mathrm  b} -{\kappa \over 2}I\big )$). Let us focus on the last term:
$ tr_0 \big ( \check r_{N0} {\mathfrak T}^{(N-k)} \hat {\mathrm b}_1 \hat {\mathfrak T}^{(N-l)} \check r_{N0}\big ) = 
 tr_0 \Big (\check r_{N0} \big ({\cal A}  + {\cal B} + {\cal C} + {\cal D} \big ) \check r_{N0}\Big )$,
where we define
\begin{eqnarray}
&&{\cal A} =\sum_{[n_{k-1}, n_1]} \prod_{1 \leq j \leq k-1}^{\leftarrow}\check r_{n_j n_j+1} \hat {\mathrm b}_1 
\sum_{[m_{l-1},  m_1]} \prod_{1 \leq j' \leq l-1}^{\rightarrow}\check r_{m_{j'} m_{j'}+1}  \label{4t}\\
&& {\cal B} =\sum_{[ n_{k-1}, n_1) } \prod_{1 \leq j \leq k-1}^{\leftarrow}\check r_{n_j n_j+1} \hat {\mathrm b}_1 
\sum_{[ m_{l-1},  m_1 ]} \prod_{1 \leq j' \leq l-1}^{\rightarrow}\check r_{m_{j'} m_{j'}+1} \nonumber \\
&&{\cal C} =\sum_{[ n_{k-1}, n_1]} \prod_{1 \leq j \leq k-1}^{\leftarrow}\check r_{n_j n_j+1} \hat {\mathrm b}_1 
\sum_{ [ m_{l-1}, m_1 )} \prod_{1 \leq j' \leq l-1}^{\rightarrow}\check r_{m_{j'} m_{j'}+1} \nonumber \\
&& {\cal D} =\sum_{[n_{k-1}, n_1)} \prod_{1 \leq j \leq k-1}^{\leftarrow}\check r_{n_j n_j+1} \hat {\mathrm b}_1 
\sum_{[m_{l-1}, m_1 )} \prod_{1 \leq j' \leq l-1}^{\rightarrow}\check r_{m_{j'} m_{j'}+1}, \nonumber
\end{eqnarray}
and $[n_{k-1}, n_j): 1\leq n_{k-1}, < ... < n_j < N-1$, and $[n_{k-1}, n_j]: 1\leq n_{k-1}, < ... < n_j = N-1$.
The last three terms above (${\cal B},\ {\cal C},\  {\cal D}$) lead to the following expressions, after using the braid relation,
 involution and the fact that $tr_0 (\check r_{N0})=I$:
\begin{eqnarray}
tr_0 \Big (\check r_{N0}  {\cal B} \check r_{N0}\Big ) = \sum_{ [n_{k-1}, n_1 ) } \prod_{1 \leq j \leq k-1}^{\leftarrow}
\check r_{n_j n_j+1} \hat {\mathrm b}_1 
\sum_{[m_{l-1}, m_2)} \prod_{2 \leq j' \leq l-1}^{\rightarrow}\check r_{m_{j'} m_{j'}+1} \nonumber\\ 
tr_0 \Big (\check r_{N0}  {\cal C} \check r_{N0}\Big ) = \sum_{[n_{k-1}, n_2 ) } \prod_{2 \leq j \leq k-1}^{\leftarrow}
\check r_{n_j n_j+1} \hat {\mathrm b}_1 
\sum_{[m_{l-1}, m_1)} \prod_{1 \leq j' \leq l-1}^{\rightarrow}\check r_{m_{j'} m_{j'}+1} \nonumber\\ 
tr_0 \Big (\check r_{N0}  {\cal D} \check r_{N0}\Big ) = {\cal N} \sum_{[n_{k-1}, n_1 ) } \prod_{1 \leq j \leq k-1}^{\leftarrow}
\check r_{n_j n_j+1} \hat {\mathrm b}_1 
\sum_{[m_{l-1}, m_1)} \prod_{1 \leq j' \leq l-1}^{\rightarrow}\check r_{m_{j'} m_{j'}+1}. \nonumber
\end{eqnarray}
The terms above clearly they depend only on $\check r_{n n+1},\ {\mathrm b}_1$. Let us now focus on the more complicated first term of (\ref{4t}), and consider:
\begin{eqnarray}
& & tr_0 \Big (\check r_{N0}  {\cal A} \check r_{N0}\Big ) =\sum_{[n_{k-1}, n_1 ) } \prod_{1 \leq j \leq k-1}^{\leftarrow}
\check r_{n_j n_j+1} \hat {\mathrm b}_1 
\sum_{[m_{l-1}, m_1)} \prod_{1 \leq j' \leq l-1}^{\rightarrow}\check r_{m_{j'} m_{j'}+1}  = \nonumber\\ 
& & \sum_{[n_{k-1}, n_1 ) }  \prod_{k'+1 \leq j \leq k-1}^{\leftarrow}
\check r_{n_j n_j+1} \hat {\mathrm b}_1\ t r_0 \Big ( \check r_{N0}\prod_{1 \leq j \leq k'}^{\leftarrow}
\check r_{n_j n_j+1}\Big |_{c_j =0, c_{k'}>0} \nonumber\\ 
& & \times \sum_{[m_{l-1}, m_1)} \prod_{1 \leq j' \leq l'}^{\rightarrow}\check r_{m_{j'} m_{j'}+1}  \check r_{N0} \Big )
\prod_{l'+1 \leq j' \leq l-1}^{\rightarrow}\check r_{m_{j'} m_{j'}+1}\Big |_{c_{j'} =0, c_{l'}>0}.   \nonumber
\end{eqnarray}

We distinguish the following cases:
\begin{enumerate}
\item {\bf$ l' = k'$}, then
\begin{eqnarray}
& & tr_0 \Big (\check r_{N0}\prod_{1 \leq j \leq k'}^{\leftarrow}
\check r_{n_j n_j+1} \prod_{1 \leq j' \leq k'}^{\rightarrow}\check r_{m_{j'} m_{j'}+1}  \check r_{N0}\Big ) = {\cal N} I^{\otimes k'}.\nonumber
\end{eqnarray}

\item {\bf $ |l' -k'| =1$}, then
\begin{eqnarray}
& & tr_0 \Big (\check r_{N0}\prod_{1 \leq j \leq k'}^{\leftarrow}
\check r_{n_j n_j+1} \prod_{1 \leq j' \leq k'}^{\rightarrow}\check r_{m_{j'} m_{j'}+1}\  \check r_{N0}\Big ) =  I^{\otimes m}.\nonumber
\end{eqnarray}
where $m = max(k', l')$

\item {\bf $ k' -l' =m+1$}, then
\begin{eqnarray}
& & tr_0 \Big (\check r_{N0}\prod_{1 \leq j \leq k'}^{\leftarrow}
\check r_{n_j n_j+1} \prod_{1 \leq j' \leq k'}^{\rightarrow}\check r_{m_{j'} m_{j'}+1}\  \check r_{N0}\Big ) = 
 \prod^{\leftarrow}_{l'+2\leq j \leq l'+m+1} \check r_{n_j n_j +1}\Big |_{c_j =0}\nonumber
\end{eqnarray}

\item {\bf $ l' -k' =m+1$}, then
\begin{eqnarray}
& & tr_0 \Big (\check r_{N0}\prod_{1 \leq j \leq k'}^{\leftarrow}
\check r_{n_j n_j+1} \prod_{1 \leq j' \leq k'}^{\rightarrow}\check r_{m_{j'} m_{j'}+1}\  \check r_{N0}\Big ) = 
 \prod^{\rightarrow}_{k'+2\leq j \leq k'+m+1} \check r_{n_j n_j +1} \Big|_{c_j =0},\nonumber
\end{eqnarray}
where we define $c_j = n_j -n_{j+1 }-1$.
\end{enumerate}

It is thus clear that the factor $tr_0 \big ( \check r_{N0} {\cal A} \check r_{N0} \big )$ is also expressed in terms of the elements 
$\check r_{n n+1}$ and ${\mathrm b}_1$.
Indeed, then all the factors ${\mathfrak t}^{(k)}$, $k \in \{1, \ldots, 2N+1\} $ are expressed in terms of  $\check r_{n n+1},\ {\mathrm b}_1$.
However, the term\\ ${\mathfrak t}^{(0)} = tr_0 \big ( \check r_{N0} \check r_{N-1 N} \ldots \check r_{12} \hat {\mathrm b}_1 
\check r_{12} \ldots \check r_{N-1 N} \check r_{N0}\big )$
can not be expressed in the general case in terms of $\check r_{n n+1},\ {\mathrm b}_1$.  Notice that in the special case where 
${\mathrm b} =I$ we obtain ${\mathfrak t}^{(0)} \propto I^{\otimes N}$

The local Hamiltonian of the system for instance  is given by the following explicit expression
\begin{equation}
{\mathfrak t}^{(2N)}=  tr_0 ({\cal T}_0^{(2N)}) = 2 \sum_{n=1}^{N-1}\check  r_{n n+1} + \hat {\mathrm b}_1 + 
2 tr_0(\check r_{N0}). 
\label{Hamopen}
\end{equation}
\end{proof}

We prove below a useful Lemma:
\begin{lemma}\label{lemma1} The elements ${\mathfrak T}^{(i)}$ and $\hat {\mathfrak T}^{(i)}$, $\  i \in \{0,\ 1\}$, introduced on Proposition \ref{Traces},
satisfy the following relations with the $A$-type Hecke algebra ${\cal H}_{N}(q=1)$ elements $\check r_{n n+1}$:
\begin{eqnarray}
&& {\mathfrak T}^{(i)} \check r_{n n+1} = \check r_{n-1 n} {\mathfrak T}^{(i)},  \qquad n \in \{2, \ldots N-1 \}\nonumber\\
&& \hat {\mathfrak T}^{(i)} \check r_{n n+1} = \check r_{n+1 n+2} \hat {\mathfrak T}^{(i)},   \qquad n \in \{1, \ldots N-2 \} \nonumber
\end{eqnarray}
\end{lemma}
\begin{proof}
The proof is straightforward for ${\mathfrak T}^{(0)},\ \hat {\mathfrak T}^{(0)}$  
due to the form of  ${\mathfrak T}^{(0)},\ \hat {\mathfrak T}^{(0)}$ and the use of the braid relation.

For ${\mathfrak T}^{(1)},\ \hat {\mathfrak T}^{(1)}$ the proof is a bit more involved. Let us focus on ${\mathfrak T}^{(1)}$ acting on $\check r_{n n+1}$, 
which can be explicitly expressed as
\begin{equation}
{\mathfrak T}^{(1)} \check r_{n n+1} = \Big ( {\mathrm A}+  {\mathrm B}+ {\mathrm C} +  {\mathrm D} \Big ) \check r_{n n+1}
\end{equation}
where we define: ${\mathrm A} = \sum_{m} \check r_{N-1 N} \ldots \check r_{m+1 m+2} \check r_{m-1 m} 
 \ldots \check r_{12}$ for  $\  m\geq n+2$ or $m\leq n-2$,
 $\ {\mathrm B}  = \check r_{N-1 N} \ldots \check r_{n n+1} \check r_{n-2 n-1} \ldots \check r_{12}$, 
$\ {\mathrm C} =  \check r_{N-1 N} \ldots \check r_{n+2 n+1} \check r_{n n+1} \ldots \check r_{12}$ and
$\ {\mathrm D} =  \check r_{N-1 N} \ldots \check r_{n+12 n+2} \check r_{n-1 n} \ldots \check r_{12}$.

Using the braid relations and the fact that $\check r^2 = I^{\otimes 2}$, we show that:
${\mathrm A} \check r_{n n+1} =  \check r_{n-1 n}  {\mathrm A}$,  
$\ {\mathrm B} \check r_{n n+1} =  \check r_{n-1 n}  {\mathrm D}$, 
$\ {\mathrm C} \check r_{n n+1} =  \check r_{n-1 n}  {\mathrm C}$  and $\ {\mathrm D} \check r_{n n+1} =  \check r_{n-1 n}  {\mathrm B}$,
which immediately lead to ${\mathfrak T}^{(1)} \check r_{n n+1} = \check r_{n-1 n} {\mathfrak T}^{(1)}$,  $\ n \in \{2, \ldots N-1 \}$. 

The proof for $\hat {\mathfrak T}^{(1)}$is in exact analogy,  so we omit the details here for brevity.
\end{proof}

For the rest of the section we focus on representations of the the $B$-type Hecke algebra ${\cal B}_N(q=1, Q=1).$

\begin{pro} \label{Proposition 9} Let $R(\lambda) =  \lambda {\mathcal P}\check r +{\mathcal P}$,  and 
$K(\lambda) =  \lambda c{\mathrm b} +I $ ($c$ is an arbitrary constant),
where $\check r$ and ${\mathrm b}$  provide a representation of the $B$-type 
Hecke algebra ${\cal B}_N(q=1, Q=1),$ and ${\mathcal P}$ is the permutation operator. 
The elements of ${\cal T}^{(i)}$, $\ i \in \{0,\ 1\}$,  introduced on Proposition \ref{Traces},
commute with the $B$-type Hecke algebra ${\cal B}_N(q=1, Q=1)$ generators:
\begin{equation}
\Big [{\cal T}_{x,y}^{(i)},\ \check r_{n n+1} \Big ] =\Big [{\cal T}_{x,y}^{(i)},\ {\mathrm b}_1 \Big ] = 0, 
~~~  n\in\{1,\, \ldots, N-1\},\ ~~~x,\ y \in X.
\end{equation}
\end{pro}
\begin{proof}
We first write down explicitly the elements ${\cal T}^{(0)}$ and  ${\cal T}^{(1)}$. Recall that
${\cal T}^{(0)}= \check r_{N0} {\mathfrak T}^{(0)} \hat {\mathrm b}_1 \hat {\mathfrak T}^{(0)} \check r_{N0},$
($ \hat {\mathrm b}= c  {\mathrm b}$) 
and from the proof of Proposition \ref{Traces}:\\
${\cal T}^{(1)} = {\mathfrak a} + {\mathfrak b} +{\mathfrak c} + {\mathfrak d} + I^{\otimes (N+1)}$, where 
${\mathfrak a} =\check r_{N0}{\mathfrak T}^{(1)} \hat {\mathrm b}_1\hat {\mathfrak T}^{(0)}\check r_{N0} $, $\ 
{\mathfrak b} =\check r_{N0}{\mathfrak T}^{(0)} \hat {\mathrm b}_1
\hat {\mathfrak T}^{(1)}\check r_{N0} $\\
${\mathfrak c} ={\mathfrak T}^{(0)} \hat {\mathrm b}_1\hat {\mathfrak T}^{(0)}\check r_{N0} $, 
$\ {\mathfrak d} =\check r_{N0}{\mathfrak T}^{(0)} \hat {\mathrm b}_1
\hat {\mathfrak T}^{(0)} $.

Using Lemma \ref{lemma1} and the expressions just above we conclude: $\big [ {\cal T}^{(i)},\ \check r_{n n+1}\big ] =0,\ 
\ n \in \{1, \ldots N-2\}$, $i \in \{0,\ 1\}$. 
Moreover, using the quadratic relation of the $B$-type algebra 
$\check r_{12} {\mathrm b}_1 \check r_{12} {\mathrm b}_1 = {\mathrm b}_1\check r_{12} {\mathrm b}_1 \check r_{12} $ 
and the form of ${\cal T}^{(0)}$ we show that $\big [{\cal T}^{(0)},\ {\mathrm b}_1 \big ]=0$, 
while use of the braid relation and the form of ${\cal T}^{(0)}$ lead to $\big [ {\cal T}^{(0)},\ \check r_{N-1 N}\big ] =0$.

It now remains to show that $\big [ {\cal T}^{(1)},\ \check r_{N-1 N}\big ] =\big [{\cal T}^{(1)},\ {\mathrm b}_1 \big ]=0$, 
the proof of the latter is more involved. Indeed, let us first focus on $\big [{\cal T}^{(1)},\ {\mathrm b}_1 \big ]$, 
it is convenient in this case to express
the first two terms of ${\cal T}^{(1)}$ as ${\mathfrak a} = {\mathfrak a}_1 + {\mathfrak a}_2$ and ${\mathfrak b} = 
{\mathfrak b}_1 + {\mathfrak b}_2$, where
${\mathfrak a}_1 =  \check r_{N0} \sum_{n=2}^{N-1}\big (\check r_{N-1 N} ...\check r_{n+1 n+2} \check r_{n-1 n}... \check r_{23} \big ) \check r_{12} \hat {\mathrm b}_1 
\hat {\mathfrak T}^{(0)} \check r_{N0}$,\\
${\mathfrak a}_2 = \check r_{N0} \check r_{N-1 N}  ... \check r_{23} \hat  {\mathrm b}_1  {\mathfrak T}^{(0)} \check r_{N0}$,\\
${\mathfrak b}_1 = \check r_{N0} {\mathfrak T}^{(0)} \hat {\mathrm b}_1 \check r_{12}\sum_{n=2}^{N-1}\big (\check r_{23} ...\check r_{n-1 n} \check r_{n+1 n+2}
 ... \check r_{N-1 N} \big ) \check r_{N0}$,\\
${\mathfrak b}_2 = \check r_{N0} {\mathfrak T}^{(0)} \hat {\mathrm b}_1 \check r_{23}  ... \check r_{N-1 N} \check r_{N0} $.\\
Using the quadratic relation $\check r_{12} {\mathrm b}_1 \check r_{12} \hat {\mathrm b}_1 = {\mathrm b}_1\check r_{12} {\mathrm b}_1 \check r_{12}$, 
and the fact that ${\mathrm b}^2 =I$ we show that: ${\mathfrak a}_1 {\mathrm b}_1 = {\mathrm b}_1 {\mathfrak a}_1 $, 
$\ {\mathfrak b}_1 {\mathrm b}_1 = {\mathrm b}_1 {\mathfrak b}_1$, $\ {\mathfrak a}_2 {\mathrm b}_1 = {\mathrm b}_1 {\mathfrak b}_2 $ and
$\ {\mathfrak b}_2 {\mathrm b}_1 = {\mathrm b}_1 {\mathfrak a}_2$, $\ {\mathfrak c} {\mathrm b}_1 = {\mathrm b}_1 {\mathfrak c}$, $\ {\mathfrak d} {\mathrm b}_1 
= {\mathrm b}_1 {\mathfrak d}$, which lead to $\big [{\cal T}^{(1)},\ {\mathrm b}_1 \big ]=0$.

We lastly focus on $\big [{\cal T}^{(1)},\ \check r_{N-1 N} \big ]$, it is convenient in this case as  well  to express
the first two terms of ${\cal T}^{(1)}$ as ${\mathfrak a} = \hat {\mathfrak a}_1 +\hat {\mathfrak a}_2$ and ${\mathfrak b} = 
\hat  {\mathfrak b}_1 + \hat  {\mathfrak b}_2$,
where we define\\
$\hat {\mathfrak a}_1 =  \check r_{N0} \check r_{N-1 N}\sum_{n=1}^{N-2}\big (\check r_{N-2 N-1} ...\check r_{n+1 n+2} 
\check r_{n-1 n}... \check r_{12} \big )  \hat  {\mathrm b}_1 
\hat {\mathfrak T}^{(0)} \check r_{N0}$,\\
$\hat{\mathfrak a}_2 = \check r_{N0} \check r_{N-2 N-1}  ... \check r_{12}  \hat {\mathrm b}_1  {\mathfrak T}^{(0)} \check r_{N0}$,\\
$\hat {\mathfrak b}_1 = \check r_{N0} {\mathfrak T}^{(0)} \hat  {\mathrm b}_1 \sum_{n=1}^{N-2}\big (\check r_{12} ...\check r_{n-1 n} \check r_{n+1 n+2}
 ... \check r_{N-2 N-1} \big ) \check r_{N-1 N}\check r_{N0}$,\\
$\hat {\mathfrak b}_2 = \check r_{N0} {\mathfrak T}^{(0)} \hat {\mathrm b}_1 \check r_{12}  ... \check r_{N-2 N-1} \check r_{N0} $.\\
Using the braid relation and the fact that $\check r^2 =I^{\otimes 2}$ we show that: $\hat   {\mathfrak a}_1 \check r_{N-1 N} =  \check r_{N-1 N} 
\hat {\mathfrak a}_1 $, 
$\ \hat  {\mathfrak b}_1\check r_{N-1 N}  = \check r_{N-1 N} \hat  {\mathfrak b}_1$, $\ \hat  {\mathfrak a}_2 \check r_{N-1 N} = \check r_{N-1 N}  
{\mathfrak c}$
and $\ \hat  {\mathfrak b}_2 \check r_{N-1 N}  =\check r_{N-1 N}   {\mathfrak d}$,  $\ {\mathfrak c} \check r_{N-1 N}  = \check r_{N-1 N} \hat  
{\mathfrak a}_2 $, 
$\ {\mathfrak d} \check r_{N-1 N}  = \check r_{N-1 N}  \hat  {\mathfrak b}_2 $,  which lead to $\big [{\cal T}^{(1)},\ \check r_{N-1 N} \big ]=0$.\\
And this concludes our proof.
\end{proof}

\begin{cor} \label{cor3}  Let $R(\lambda) =  \lambda {\mathcal P}\check r +{\mathcal P}$,  and 
$K(\lambda) =  \lambda c{\mathrm b} +I $ ($c$ is an arbitrary constant),
where $\check r$ and ${\mathrm b}$  provide a representation of the $B$-type 
Hecke algebra ${\cal B}_N(q=1, Q=1),$ and ${\mathcal P}$ is the permutation operator. 
Let also ${\mathfrak t}^{(k)}$, $k \in \{1, \ldots, 2N+1\}$ be the mutually commuting 
charges as defined in Proposition \ref{Traces}, and $tr_0(\check r_{N0}) \propto I^{\otimes 2}$, then
\begin{equation}
\Big [{\mathfrak t}^{(k)},\  {\cal T}_{x,y}^{(i)} \Big ] =0, \qquad i \in \{0,\ 1\}.
\end{equation}
\end{cor}
\begin{proof}
The proof is straightforward, based on Propositions \ref{Traces} and \ref{Proposition 9}. 
\end{proof}

\begin{cor} \label{cor4}  Let $R(\lambda) =  \lambda {\mathcal P}\check r +{\mathcal P}$,  and 
$K(\lambda) =  \lambda c{\mathrm b} +I$ ($c$ is an arbitrary constant),
where $\check r$ and ${\mathrm b}$  provide a representation of the $B$-type 
Hecke algebra ${\cal B}_N(q=1, Q=1),$ and ${\mathcal P}$ is the permutation operator. Let also ${\mathfrak t}^{(k)}$, $k \in \{0, \ldots, 2N+1\}$ be the mutually commuting 
charges as defined in Proposition \ref{Traces}, and $tr_0(\check r_{N0}) \propto I^{\otimes 2}$. In the special case ${\mathrm b} = I$:
\begin{equation}
\Big [{\mathfrak t}^{(k)},\  {\cal T}_{x,y}^{(1)} \Big ] =0 \Rightarrow \Big [ {\mathfrak t}(\lambda),\ {\cal T}_{x,y}^{(1)} \Big ]=0.
\end{equation}
\end{cor}
\begin{proof}
The proof follows directly from  Propositions \ref{Traces} and \ref{Proposition 9}, 
and the fact that for ${\mathrm b} = I$, ${\cal T}^{(0)} \propto I^{\otimes (N+1)}$ and ${\mathfrak t}^{(0)} \propto I^{\otimes N}$.
\end{proof}

\begin{rem}\label{remarktw}The twisted co-products for the finite algebra generated by the element of ${\cal T}^{(1)}$,
 in the special case ${\mathrm b} =I$ can be expressed as follows, after recalling the notation introduced in the proof of Proposition \ref{Traces}:
\begin{equation}
{\cal T}^{(1)} = 2 c \sum_{n=1}^N \big ( r_{0N} r_{0N-1} \ldots r_{0n+1} \check r_{n0} \hat r_{0n+1} 
\ldots \hat r_{0N-1} \hat r_{0N}\big ), \label{tta1}
\end{equation} 
where 
$r = {\cal P} \check r $, $\hat r = \check  r {\cal P}$  and ${\cal P}$ the permutation operator.
After using expression (\ref{tta1}),  the brace relation and recalling
 that $r ={\cal P} \check r,\ \hat r = \check r {\cal P}$, we have
\begin{equation}
{\cal T}^{(1)}=  2c \sum_{n=1}^N \big ( \check r_{n n+1}\check r_{n+1 n+2} \ldots \check r_{N-1 N} \check r_{N0} \check r_{N-1 N} 
\ldots \check r_{n+1 n+2} \check  r_{n n+1} \big ).   \label{cotw2}
\end{equation}
\end{rem}
Note that explicit expressions of the above co-products for  (\ref{cotw2})  can be computed for the brace solution.
We shall derive in the next subsection the co-products associated to Lyubashenko's solutions recovering the
twisted  co-products of Corollary \ref{prop2}.

An interesting direction to pursue is the derivation of
analogous results in the case of the twisted algebras extending the findings of \cite{CraDoi} on the duality between
twisted Yangians and Brauer algebras \cite{Molev} to include set-theoretic solutions. We aim at examining  whether
the corresponding transfer matrix can be expressed in terms of the elements of the Brauer algebra, and
also check if the elements of the Brauer algebra commute with a finite sub-algebra of the twisted algebra.
These findings will have significant  implications on the symmetries of open transfer matrices providing
valuable information on their spectrum.

 \subsection{More examples of  symmetries}

In this subsection we present examples of symmetries of the double row transfer matrix partly inspired 
by the symmetries in \cite{DoiSmo}, but also some new ones.
 Let $(X, \check r)$ be a set-theoretic solution, as usually we denote $\check r(x,y)=(\sigma _{x}(y), \tau _{y}(x))$. 
In all the examples below we assume that the solution $(X, {\check r})$ is involutive, non-degenerate and finite. 
Also, we always assume that $\hat K=I$ in (\ref{transfer}).

The following class of symmetries is similar to those  of Proposition 4.6 in \cite{DoiSmo}.

\begin{lemma}  
Let $(X, \check r)$ be a set-theoretic solution of the braid equation and let $f : X \rightarrow X$ be 
an isomorphism of solutions, so $f(\sigma _{x}(y)) = \sigma _{f(x)}(f(y))$ and
 $f(\tau_{x}(y)) = \tau _{f(x)}(f(y))$. Denote $M=\sum_{x\in X}e_{x,f(x)}$, and
let ${\mathfrak t}(\lambda )$ be  the double row transfer matrix for $R(\lambda ) = 
{\cal P} + \lambda {\cal P}{\check r }$ and $K(\lambda)= \lambda c {\mathrm b}+I$, 
where $c$ is an arbitrary constant and ${\mathrm b} = \sum_{x\in X} e_{x, k(x)}$.
Then, given that $f(k(x))= k(f(x))$: \[\Big [M^{\otimes N},\  {\mathfrak t}(\lambda )\Big ]= 0.\]
\end{lemma}
\begin{proof} 
Notice that $M\otimes M$ commutes with $r = {\cal P}{\check r}$ and $\hat r = {\check r} {\cal P}$, 
which leads to 
$ M^{\otimes (N+1)} T(\lambda)= T(\lambda)  M^{\otimes (N+1)} $ and 
$M^{\otimes (N+1)} \hat T(\lambda)= \hat T(\lambda)  M^{\otimes (N+1)}$, 
also due to $f(k(x))= k(f(x))$ we have that ${\mathrm b} M = M {\mathrm b}$. 
These commutation relations then  lead to $\big [{\cal T}(\lambda),\   M^{\otimes( N+1)}\big  ]=0$, and from the latter we obtain,
 following the proof of Proposition 4.6 in \cite{DoiSmo}, $M^{\otimes N} {\cal T}_{f(x), f(x)} = {\cal T}_{x,x} M^{\otimes N}$,
which directly leads to $\big [{\mathfrak t}(\lambda),\   M^{\otimes N}\big  ]=0$.
\end{proof}

The following Lemma also follows from Proposition 4.9 in \cite{DoiSmo}.

\begin{lemma}  Let $ (X, \check r)$ be a finite, non degenerate involutive set-theoretic solution of the braid equation.
 Let $x_{1},\ldots ,x_{\alpha }\in X$ for some $\alpha \in \{1,\ldots, {\cal N}\}$ and assume that 
${\check  r}(x_{i},y) = (y,x_{i})$, $\forall y \in  X$.
 Then, $\forall i,j\in \{1, 2, \ldots , \alpha \}$
 \[\Big [ \Delta ^{(N)}(e_{x_{i},x_{j}}),\  {\mathfrak t}(\lambda )\Big ]= 0,\]
where $ \Delta^{(N)}(e_{x_i,x_j}) =\sum_{n=1}^N I \otimes \ldots \otimes \underbrace{ e_{x_i, x_j}}_{n^{th}\ position}
\otimes \ldots \otimes I,$ ${\mathfrak t}(\lambda )$ is the double row transfer matrix for $R(\lambda ) = {\cal P} + \lambda {\cal P}{\check r }$ and
$K(\lambda)$ such that $\big [ K(\lambda),\  e_{x_i, x_j}\big] =0$.
\end{lemma}
\begin{proof} 
The co-product $\Delta(e_{x_i, x_j})$ commutes with both $r={\cal P} \check r$ and  $\hat r = \check r{\cal P}$,
then  as in the proof of Proposition 4.9 in \cite{DoiSmo} it can be shown that
$\big [\Delta ^{(N+1)}(e_{x_{i},x_{j}}),\ T(\lambda)\big ]  = \big [\Delta ^{(N+1)}(e_{x_{i},x_{j}}),\ \hat T(\lambda)\big ]  = 0$,
recall also that $\big [K(\lambda),\  e_{x_i, x_j}\big] =0$. The three commutation relations then immediately  lead to 
$\big [\Delta ^{(N+1)}(e_{x_{i},x_{j}}),\ {\cal T}(\lambda)\big ]   =0$. Then following  the proof of  Proposition 4.9 in \cite{DoiSmo}, 
we focus on the diagonal entries of the latter commutator:
$\big[ \Delta^{(N)}(e_{x_i,x_j}),\  {\cal T}_{x_i,x_i}(\lambda)\big ]  = - {\cal T}_{x_j,x_i}(\lambda) + \delta_{ij} {\cal T}_{x_j,x_i}(\lambda)$,
 $\big[ \Delta^{(N)}(e_{x_i,x_j}), \  {\cal T}_{x_j,x_j}(\lambda)\big ]  =  {\cal T}_{x_j,x_i}(\lambda) -  \delta_{ij} {\cal T}_{x_j,x_i}(\lambda)$ and
$\big[ \Delta^{(N)}(e_{x_i, x_j}),\   {\cal T}_{z,z}(\lambda)\big ]   =0, \ z \neq x_i,\ x_j$, and 
we conclude that $\big [ \Delta ^{(N)}(e_{x_{i},x_{j}}),\  {\mathfrak t}(\lambda )\big ]= 0$.
\end{proof}

The following Lemma is similar to Proposition $4.11$ in \cite{DoiSmo},  but here for the double row transfer matrix, 
we obtain a stronger result:

\begin{lemma}  Let $ (X, \check r)$ be a finite, non degenerate involutive set-theoretic solution of the braid equation.
 Let $x_{1},\ldots ,x_{\alpha }\in X$ for some $\alpha \in \{1,\ldots ,{\cal N}\}$ and assume that 
${\check  r}(x_{i},x_{i}) = (x_{i},x_{i})$.
 Then, $\forall i,j\in \{1, 2, \ldots , \alpha \}$
 \[\Big [ e_{x_{i},x_{j}}^{\otimes N},\ {\mathfrak t}(\lambda )\Big ]= 0\]
where ${\mathfrak t}(\lambda )$ is the double row transfer matrix for $R(\lambda ) =  {\cal P} + \lambda  {\cal P}{\check r }$  and $K(\lambda)\propto I$.
\end{lemma}
\begin{proof} Similarly as in the proof of Proposition 4.11 from \cite{DoiSmo} it can be shown that $e_{x_{i},x_{j}}^{\otimes N}$ commutes with 
$\check r_{n n+1}$,  $\forall n\in \{ 1, \ldots, N-1\}$. The result now immediately follows from Proposition \ref{Traces} 
and from the fact that for ${\mathrm  b} = I$, we have 
$\mathcal  {T}^{(0)} = I^{\otimes (N+1)}$ and ${\mathfrak t}^{(0)} = I^{\otimes N}$. 
\end{proof}

We also present the following new examples of symmetries, different to  the ones derived in \cite{DoiSmo}.
Let us first introduce some invariant subsets of a set-theoretic solution.
Let $(X, {\check r})$ be an involutive, non-degenerate set-theoretic solution.

\begin{defn}
 Let $(X,{\check r})$ be a finite  set-theoretic solution of the braid equation and let $Y\subseteq X$. 
Denote ${\check r}(x,y)=(\sigma _{x}(y), \tau _{y}(x))$. 
We say that  $Y$ is a $\sigma$-equivariant set  if whenever $x,y\in Y$ then $\sigma _{x}(y)$ and $\tau _{y}(x)\in Y$.
\end{defn}

\begin{pro}\label{A1} Let $(X,{\check r})$ be an involutive 
non-degenerate solution of the braid equation.   Let $Y, Z\subseteq X$ be $\sigma$-equivariant sets. 
Define $M_{Y,Z}=\sum_{i\in Y, j\in Z}e_{i,j}$, 
then  \[\Big [M_{Y, Z}^{\otimes N},\ {\mathfrak  t}(\lambda )\Big ]= 0\]
 where ${\mathfrak t}(\lambda )$ is the double row transfer matrix for $K(\lambda) \propto I$ and $R(\lambda)={\cal P}+ \lambda {\cal P} \check r$.
\end{pro}
\begin{proof} By Proposition \ref{Traces} it suffices to show that  $M_{Y,Z}$ commutes with  $\check r_{n n+1}$, $\forall   n\in \{ 1, \ldots, N-1\}$.
 %Notice $M_{Y, Z}=M^{\otimes N}$ where $M=\sum_{i\in Y, j\in Z}e_{i,j}$.
Observe first that \[M\otimes M=\sum_{i,j\in Y, k,l\in Z}e_{i,k}\otimes e_{j,l}.\]
Also,  ${\check r}(M\otimes M)=M\otimes M$ and $(M\otimes M){\check r}=M\otimes M$,
\[(M\otimes M){\check r}=\sum_{i,j\in Y, k,l\in Z}e_{i,k}\otimes e_{j,l} \sum_{x,y\in X}e_{x, \sigma _{x}(y)}\otimes e_{y, \tau _{y}(x)}=\]
\[=
\sum_{i,j\in Y, k,l\in Z}e_{i, \sigma _{ k}(l)}\otimes e_{j, \tau _{l}(k)}=M\otimes M,  \]
 because mapping ${\check r }:Y\otimes Y\rightarrow Y\otimes Y $  with
$(k,l)\rightarrow (\sigma _{k}(l), \tau _{l}(k))$ is bijective (as explained in the end of the proof).

To show that ${\check r}(M\otimes M)=M\otimes M$ observe that, because $\check r$ is involutive it follows that 
$\check r=\sum_{x,y\in X} e_{\sigma _{x}(y),x}\otimes e_{\tau _{y}(x), y}$.
Therefore,  \[{\check r}(M\otimes M)=\sum_{x,y\in X} e_{\sigma _{x}(y),x}
\otimes e_{\tau _{y}(x), y}\sum_{i,j\in Y, k,l\in Z}e_{\sigma _{i}(j),k}\otimes e_{\tau _{j}(i),l}=M\otimes M, \]
because $\check r:Z\otimes Z\rightarrow Z\otimes Z$ is a bijective function.

%In a similar way it can be proved that ${\check r}_{2,1}(M\otimes M)=M\otimes M=(M\otimes M){\check r}_{2,1}$ indeed it follows 
%from the fact that $r_{2,1}: Y\otimes Y\rightarrow Y\otimes Y$ is bijective, because $\check r$ is bijective. 
Therefore $M^{\otimes N}$ commutes with $ \check r_{n n+1}$,  $\forall n \in \{1, \ldots, N-1\}$.
The result now follows from Proposition \ref{Traces}.

To show that $\check r$ is a  bijective function on $Y\times Y$ observe that 
 $\check r$ has the zero kernel on $X\otimes X$, so is injective on $Y\otimes Y$. Notice that 
 $\check r (Y\otimes Y)\subseteq Y\otimes Y$ since $Y$ is $\sigma$-equivariant set. Because 
$\check r: Y\otimes Y\rightarrow Y\otimes Y$ is injective then 
 $\check r(Y\otimes Y)$ has the same cardinality as $Y\otimes Y$, hence $\check r: Y\otimes Y\rightarrow Y\otimes Y$
 is surjective and hence bijective. 
%Similarly it is shown that ${\check r}_{21}$ is bijective on $Y\otimes Y$.
\end{proof}

\begin{rem} We we choose $\sigma $-equivariant subsets of $X$ which have pairwise empty intersections 
we get similar algebra of symmetries as in the previous Lemma.
\end{rem}

\begin{defn}
Let $z\in X$. By the {\em orbit } of $z$ we will mean the smallest set $Y\subseteq X$ 
such that $z\in Y$ and $\sigma _{x}(y)\in Y$ and $\tau _{x}(y)\in Y$, for all $y\in Y, x\in X$.
\end{defn}
We have also the following symmetries:

\begin{lemma} Let $(X, \check r)$ be an involutive, non degenerate solution of the braid equation and let $Q_{1}, \ldots , Q_{t}$ 
be orbits of $X$. \\ Define $W_{p_{1}, \ldots , p_{t},q_{1}, \ldots , q_{t}}=\{i_{1}, i_{2}, \ldots , i_{n}, j_{1}, j_{2},  \ldots , j_{n}:$
exactly $p_{i}$ elements among $ i_{1}, i_{2}, \ldots , i_{n}$ belong to the orbit $Q_{i}$ and
 exactly $q_{i}$ elements among $ j_{1}, j_{2}, \ldots , j_{n}$ belong to the orbit $Q_{i}$ for every $i\leq t\}$.

 Fix non-negative integers $p_{1}, \ldots , p_{t}, q_{1}, \ldots , q_{t}$, and define  
\[A_{p_{1}, \ldots , p_{t}, q_{1}, \ldots , q_{t}}=\sum_{i_{1}, i_{2}, \ldots , i_{n}, j_{1}, j_{2},  \ldots , j_{n\in W_{p_{1}, 
\ldots , p_{t}, q_{1}, \ldots , q_{t}} }}e_{i_{1}, j_{1}}\otimes e_{i_{2}, j_{2}}\otimes \cdots \otimes e_{i_{n}, j_{n}}.\]
 Then $A_{p_{1}, \ldots , p_{t}, q_{1}, \ldots , q_{t}}$ commutes with $\check r_{n, n+1}$ and so it  commutes with 
 the double row transfer matrix when $K(\lambda) \propto I$.
\end{lemma}

\begin{proof} It follows from the fact that if $(X,\check r)$ is an involutive, non-degenerate set-theoretic solution of the Braid equation then 
 $r(Q_{i}, Q_{j})\subseteq (Q_{i}, Q_{j})$ and it is a bijective map, for every $i, j\leq t$.
\end{proof}

\subsection{Symmetries associated to  Lyubashenko's solution.}
We focus in this subsection on the symmetries of the open transfer matrix constructed 
from  the Lyubashenko solution of  Proposition \ref{prop1}. 

\begin{cor} \label{Proposition 10} Let  ${\mathfrak t}(\lambda)$  be the double row transfer matrix for 
$R(\lambda) =\lambda {\cal P} \check r + {\cal P},$ where $\check r$ is Lyubashenko's  solution of  Proposition \ref{prop1},
and $K(\lambda) = \lambda c {\mathrm b}+ I,$ $c$ is an arbitrary constant and ${\mathrm b}$ satisfies: ${\mathrm b}^2 =I$ 
and ${\mathrm b}_1 \check r_{12} {\mathrm b}_1 \check r_{12} = \check r_{12} {\mathrm b}_1  \check r_{12}{\mathrm b}_1  $. Then,
\begin{equation}
\Big [{\mathfrak t}(\lambda),\  {\cal T}_{x,y}^{(1)}\Big ] = 0, ~~~x, y \in X.\label{symm3}
\end{equation}
\end{cor}
\begin{proof}
Recall from the notation introduced 
in the proof of Proposition \ref{Traces} that ${\cal T}^{(0)} = r_{0N} ... r_{01} c {\mathrm b}_0 \hat r_{01} ... 
\hat r_{0N}$, $\ \hat r = {\cal P} r {\cal P}$. 
Then using  the fact that in the special case of Lyubashenko's solutions, $r_{0n} = {\mathbb V}_0^{-1}{\mathbb V}_n $,
we can explicitly write ${\cal T}^{(0)}=  {\mathbb V}_0^{-N} c {\mathrm b}_0 {\mathbb V}_0^{N}$, which is a $c$-number matrix and 
${\mathfrak t}^{(0)} =   tr_0 (c {\mathrm b}_0)$,
which immediately leads to $\big [ {\mathfrak t}^{(0)},\ {\cal T}^{(1)}_{x,y}\big ] =0$, and via Propositions \ref{Traces} and \ref{Proposition 9} 
we arrive at (\ref{symm3}).
\end{proof}

%\noindent {\bf Remark 6.} {\em  Proposition 6 holds  for the double row transfer matrix as well 
%in the special case $\hat K =K =\mbox{id}$.}\\
%The proof is straightforward in accordance to the one of Proposition 6.

\begin{cor} \label{Proposition 11} 
Let  $\check r$ be Lyubashenko's  
solution of  Proposition \ref{prop1}, $R(\lambda) =\lambda {\cal P} \check r + {\cal P},$  and $K(\lambda)\propto I$. 
Then the elements ${\cal T}_{x,y}^{(1)}$  of  (\ref{tta1}) are twisted co-products of $\mathfrak{gl}_{\cal N}$, 
and hence the corresponding double row 
transfer matrix ${\mathfrak t}(\lambda)$  is $\mathfrak{gl}_{\cal N} $ symmetric.
\end{cor}
\begin{proof}
 Recall that in the special case where ${\mathrm b} =I$ the quantity  ${\cal T}^{(1)}$ is given in (\ref{tta1}).
In the case of the special solutions of Proposition \ref{prop1}
recall that $\check r_{n0} =  {\mathbb V}_0^{-1}{\cal P}_{0n} {\mathbb V}_0$, then expression (\ref{tta1}) simplifies to
\begin{equation}
{\cal  T}^{(1)} = 2c\sum_{n=1}^N \big ({\mathbb V}_{n}^{(N-n+1)} {\cal P}_{0n}  {\mathbb V}_{n}^{-(N-n+1)}\big ).  \label{tta2} \nonumber
\end{equation}
Recall also from Proposition \ref{prop1} that ${\mathbb V} = \sum_{x \in X} e_{\sigma(x), x}$ and ${\cal P} = 
\sum_{x, y \in X} e_{x,y} \otimes e_{y,x}$, then ${\cal T}^{(1)}$ can be explicitly expressed as (we set for simplicity $2c=1$)
\begin{equation}
{\cal  T}^{(1)} = \sum_{x, y\in X} e_{x, y} \otimes \big (\sum_{n=1}^N I \otimes \ldots \otimes \underbrace{e_{\sigma^{N-n+1}(y), 
\sigma^{N-n+1}(x)}}_{n^{th} position}\otimes \ldots \otimes I \big). \nonumber
\end{equation}
The latter expression immediately provides the elements ${\cal T}_{x,y}^{(1)} = \Delta_1^{(N)}(e_{\sigma(y), \sigma(x)})$,
where the twisted $N$ co-product $\Delta^{(N)}_1$  of $\mathfrak{gl}_{\cal N}$ is defined in Corollary \ref{prop2}, expression (\ref{delta1}). Then due to Corollary \ref{cor4} we deduce that the corresponding  double row transfer matrix  ${\mathfrak t}(\lambda)$  is 
$\mathfrak{gl}_{\cal N} $ symmetric
And with this we conclude our proof (compare also with the results in Corollary \ref{Proposition 7} for ${\mathbb K}^{(0)} = I$).
\end{proof}
%\noindent {\bf Remark 6.} {\em  Proposition 6 holds  for the double row transfer matrix as well 
%in the special case $\hat K =K =\mbox{id}$.}\\
%The proof is straightforward in accordance to the one of Proposition 6.

\begin{cor}
Let $\check r$ be Lyubashenko's solution of Proposition \ref{prop1}, and 
$M=\sum_{y \in X}\alpha_{y} M_{y},$ where $M_y = \sum_{x\in X} e_{x, \sigma^y(x)}$ and $\alpha_y \in {\mathbb C}$.  Let also
${\mathfrak t}(\lambda )$ be the double row transfer matrix for $R(\lambda)={\cal P}+ \lambda {\cal P} \check r$ 
and $K= \lambda c {\mathrm b}+I$, 
where $c$ is an arbitrary constant and ${\mathrm b} = \sum_{x\in X} e_{x, k(x)}$, then  \begin{equation}
\Big [ M_y^{\otimes N},\  {\mathfrak t}(\lambda) \Big ]=\Big [ M^{\otimes N},\  {\mathfrak t}(\lambda) \Big ]= 0, \label{lc1}
\end{equation}
provided that $\sigma^y(k(x))= k(\sigma^y(x))$.

Moreover, let $\xi \in \mathbb C$ and $A=\sum_{x\in X} \xi ^{x}e_{x,x}$, then
\begin{equation}
\Big [ A^{\otimes N},\  {\mathfrak t}(\lambda) \Big ]= 0, \label{lc2}
\end{equation}
provided that $\xi^x= \xi^{k(x)}$ and $\xi^{x+y} = \xi^{\sigma(y) + \tau(x)}$.
\end{cor}
\begin{proof}
Observe that 
\begin{equation}
{\check r}\big (e_{z,w}\otimes e_{\hat z, \hat w}\big )=\big (e_{\sigma(\hat z), \sigma(\hat w)}\otimes e_{\tau(z),\tau(w)}\big ){\check r}, \label{obs}
\end{equation}
then
${\check r }(M_{y}\otimes M_{\hat y})=(M_{\hat y}\otimes M_{y}){\check r}$, hence  $\check r$ as well as $r ={\cal P}\check r$
and $\hat r = \check r {\cal P}$ commute with $M_y \otimes M_y$ and $M \otimes M$. Moreover, $M_y,\ M$ commute with $K(\lambda)$ due to 
$\sigma^y(k(x))= k(\sigma^y(x))$. $T,\ \hat T$ and $K$ commute with $M_y^{\otimes (N+1)}$ and $M^{\otimes (N+1)},$ and consequently so does the double row transfer matrix ${\cal T}(\lambda)$.
 From $\big [ M_y^{\otimes (N+1)},\ {\cal T}(\lambda)\big] =0$ we obtain
$e_{x, \sigma^y(x)} \otimes M^{\otimes N}_y {\cal T}_{\sigma^y(x), \sigma^y(x)}=
e_{x, \sigma^y(x)} \otimes  {\cal T}_{x, x}M^{\otimes N}_y $, then\\  $\sum_{x\in X} M^{\otimes N}_y {\cal T}_{\sigma^y(x), \sigma^y(x)} =\sum_{x\in x} {\cal T}_{x, x}M^{\otimes N}_y,$ similarly for $M,$ which lead to (\ref{lc1}).

Notice that in the special case ${\mathrm b}  = I$ (\ref{lc1}) follows also immediately from the fact that 
$M_y^{\otimes N}$ and $M^{\otimes N}$ commute with $\check r_{n n+1}$, $\forall n \in \{1, \ldots, N-1\}$ and Proposition \ref{Traces}.

Similarly, via (\ref{obs}) and the fact that $\xi^{x+y} = \xi^{\sigma(y) + \tau(x)}$ we show that\\ $\big [A \otimes A,\ \check r \big]=0,$ and hence
 $\big [T(\lambda),\ A^{\otimes (N+1)}\big ] = \big [\hat T(\lambda),\ A^{\otimes (N+1)}\big ] =0$. Moreover, due to $\xi^x= \xi^{k(x)}$
 we show that $\big [K(\lambda),\ A \big ] =0$, and consequently $\big [ {\cal T}(\lambda),\ A^{\otimes( N+1)}\big ] =0$. By taking the 
trace over the auxiliary space in the latter commutator we arrive at (\ref{lc2}).
\end{proof}

\subsection*{Acknowledgments}

\noindent  We are thankful to  Robert Weston for useful discussions.
AD acknowledges support from the EPSRC research grant EP/R009465/1 and 
AS acknowledges support from the EPSRC research grant EP/R034826/1.
The authors also acknowledge support from the EPSRC research grant  EP/V008129/1.

\end{document}